\pdfoutput=1
\documentclass{sig-alternate-05-2015}

\usepackage{bm}
\usepackage{xcolor}
\usepackage{enumitem}
\usepackage{caption}
\usepackage{subcaption}
\usepackage{graphicx}
\usepackage{booktabs}
\usepackage{tabularx}
\usepackage{url}
\usepackage{hyperref}
\hypersetup{colorlinks,citecolor=blue}
\usepackage{algorithmicx}
\usepackage[ruled,vlined]{myalgorithm2e}
\usepackage{times}

\usepackage{amsmath,amsfonts,amssymb,amsthm}

\usepackage{tikz,nicefrac,amsmath,pifont,tkz-graph,makecell,tkz-graph}
\usetikzlibrary{arrows,backgrounds,patterns,matrix,shapes,fit,calc,shadows,plotmarks}

\newtheorem*{theorem*}{Theorem}
\newtheorem*{definition*}{Definition}

\newcommand{\hide}[1]{}
\newcommand{\xhdr}[1]{\vspace{0.75mm}\noindent{{\bf #1.}}}
\newcommand{\triangleenum}{\textnormal{TriEnum}}

\newcommand{\pluseq}{\mathrel{+}=}
\newcommand{\minuseq}{\mathrel{-}=}

\newcommand{\counts}{\textnormal{counts}}
\newcommand{\prefix}{\textnormal{prefix}}
\newcommand{\suffix}{\textnormal{suffix}}

\newcommand{\concat}{\textnormal{concat}}
\newcommand{\keys}{\textnormal{keys}}
\newcommand{\reverse}{\textnormal{reverse}}

\newcommand{\tstart}{\textnormal{start}}
\newcommand{\tend}{\textnormal{end}}

\newcommand{\dataset}[1]{\textsc{#1}}
\newcommand{\emaileu}{\dataset{Email-Eu}}

\newcommand{\phone}{\dataset{Phonecall-Eu}}

\newcommand{\sms}{\dataset{SMS-A}}
\newcommand{\messages}{\dataset{CollegeMsg}}
\newcommand{\stackoverflow}{\dataset{StackOverflow}}

\newcommand{\fbwall}{\dataset{FBWall}}
\newcommand{\bitcoin}{\dataset{Bitcoin}}
\newcommand{\wikitalk}{\dataset{WikiTalk}}
\newcommand{\phonecallme}{\dataset{Phonecall-ME}}
\newcommand{\smsme}{\dataset{SMS-ME}}

\newcommand{\enumtopsep}{2pt}

\makeatletter
\def\thm@space@setup{%
  \thm@preskip=1pt \thm@postskip=1pt
}
\makeatother

\newcommand{\prenodecount}{\textnormal{pre\_nodes}}
\newcommand{\postnodecount}{\textnormal{post\_nodes}}
\newcommand{\middlesum}{\textnormal{mid\_sum}}
\newcommand{\presum}{\textnormal{pre\_sum}}
\newcommand{\postsum}{\textnormal{post\_sum}}
\newcommand{\nodecount}{\textnormal{node\_count}}
\newcommand{\sumtxt}{\textnormal{sum}}
\newcommand{\nbr}{\textnormal{nbr}}
\newcommand{\dir}{\textnormal{dir}}
\newcommand{\globalpre}{\textnormal{count\_pre}}
\newcommand{\globalpost}{\textnormal{count\_post}}
\newcommand{\globalmid}{\textnormal{count\_mid}}
\newcommand{\globalall}{\textnormal{count}}
\newcommand{\uorv}{\textnormal{uorv}}
\newcommand{\muorv}{\textnormal{1-uorv}}
\newcommand{\utov}{\textnormal{utov}}

\begin{document}
\sloppy

\CopyrightYear{2017} \setcopyright{acmcopyright}
\conferenceinfo{WSDM '17,}{February 6--10, 2017, Cambridge, United Kingdom.}
\isbn{978-1-4503-4675-7/17/02}\acmPrice{\$15.00}
\doi{http://dx.doi.org/10.1145/3018661.3018731}

\captionsetup{labelfont=bf,font=bf,skip=1pt}

\clubpenalty=10000
\widowpenalty = 10000

\title{Motifs in Temporal Networks}
%
%
%
%
%

\numberofauthors{3}
\author{
%
%
\alignauthor
Ashwin Paranjape\titlenote{These authors contributed equally to this work.} \\
       \affaddr{Stanford University}\\
       \email{ashwinpp@stanford.edu}
\alignauthor
Austin R.~Benson\raisebox{9pt}{$\ast$} \\
       \affaddr{Stanford University}\\
       \email{arbenson@stanford.edu}
\alignauthor
Jure Leskovec \\
       \affaddr{Stanford University}\\
       \email{jure@cs.stanford.edu}
}

\maketitle

%
%



\begin{abstract}

Networks are a fundamental tool for modeling complex systems in a variety of
domains including social and communication networks as well as biology and
neuroscience.  Small subgraph patterns in networks, called network
motifs, are crucial to understanding the structure and function of these
systems.  However, the role of network motifs in temporal networks, which
contain many timestamped links between the nodes, is not yet well understood.

Here we develop a notion of a temporal network motif as an elementary unit of
temporal networks and provide a general methodology for counting such motifs.
We define temporal network motifs as induced subgraphs on sequences of temporal edges,
design fast algorithms for counting temporal motifs, and prove
their runtime complexity.  Our fast algorithms achieve up to
56.5x speedup compared to a baseline method.
Furthermore, we use our algorithms to count temporal motifs in a variety of
networks.  Results show that networks from different domains have
significantly different motif counts, whereas networks from the same domain
tend to have similar motif counts.  We also find that different motifs occur 
at different time scales, which provides further insights into structure and function of temporal networks. 

\end{abstract}

\section{Introduction}
\label{sec:introduction}

Networks provide an abstraction for studying complex systems in a broad set of
disciplines, ranging from social and communication networks to molecular biology
and neuroscience~\cite{newman2003structure}.  Typically, these systems are
modeled as static graphs that describe relationships between objects (nodes) and
links between the objects (edges). However, many systems are not static as the
links between objects dynamically change over time~\cite{holme2012temporal}.
Such {\em temporal networks} can be represented by a series of timestamped
edges, or \emph{temporal edges}.  For example, a network of email or instant
message communication can be represented as a sequence of timestamped directed
edges, one for every message that is sent from one person to another.  Similar
representations can be used to model computer networks, phone calls, financial
transactions, and biological signaling networks.

\begin{figure}[!t]
\centering
\begin{subfigure}[b]{0.3\textwidth}\phantomcaption\label{fig:intro1A}\end{subfigure}
\begin{subfigure}[b]{0.3\textwidth}\phantomcaption\label{fig:intro1B}\end{subfigure}
\begin{subfigure}[b]{0.3\textwidth}\phantomcaption\label{fig:intro1C}\end{subfigure}
\includegraphics[width=\columnwidth]{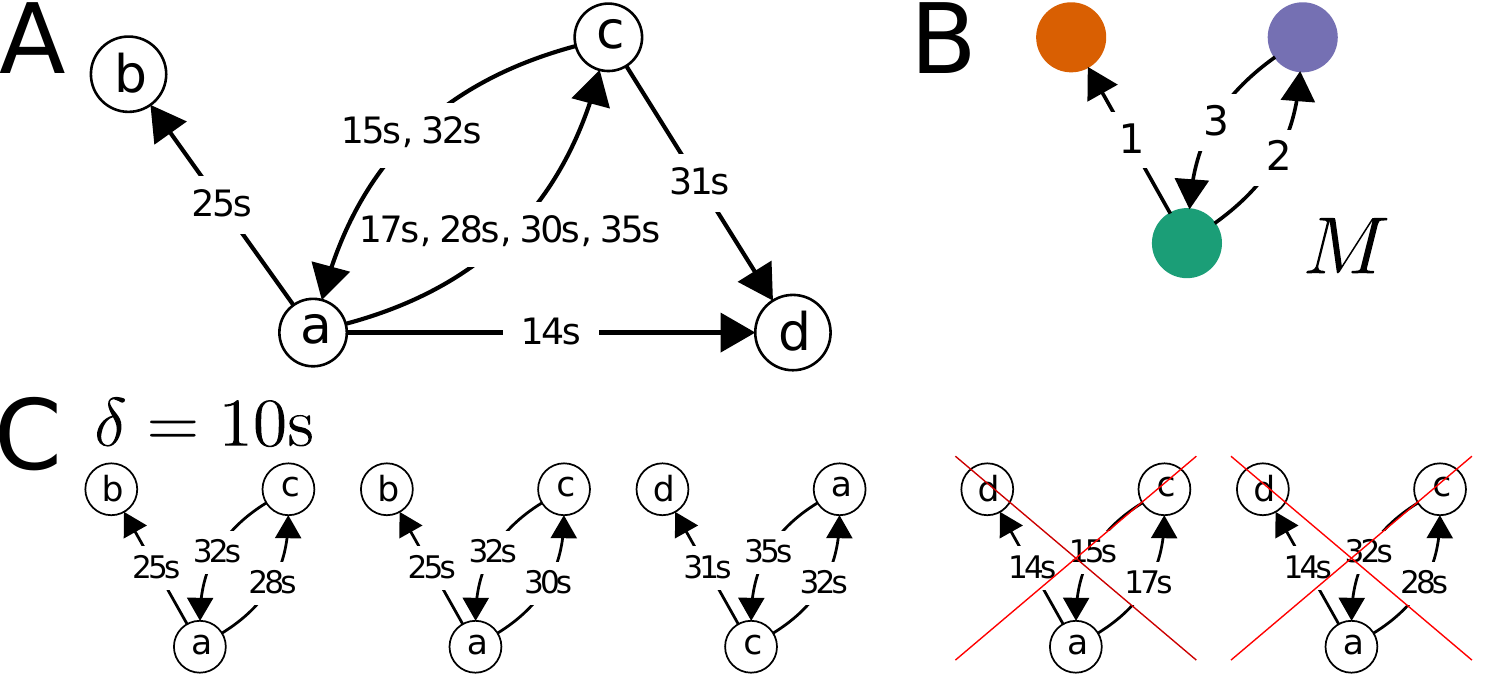}
\caption{%
Temporal graphs and $\delta$-temporal motifs.
\textnormal{A}: A temporal graph with nine temporal edges.  Each edge
has a timestamp (listed here in seconds).
\textnormal{B}: Example $3$-node, $3$-edge $\delta$-temporal motif $M$.  The
edge labels correspond to the ordering of the edges.
\textnormal{C}:
Instances of the $\delta$-temporal motif $M$ in the graph for $\delta$ = 10
seconds.  The crossed-out patterns are not instances of $M$ because either the
edge sequence is out of order or the edges do not all occur within the time
window $\delta$.
\vspace{-0.7cm}
}
\label{fig:intro1}
\end{figure}

While such temporal networks are ubiquitous, there are few tools for modeling
and characterizing the underlying structure of such dynamical systems.  Existing
methods either model the networks as strictly growing where a pair of nodes
connect once and stay connected
forever~\cite{barabasi1999emergence,jacobs2015assembling,leskovec2007graph} or
aggregate temporal information into a sequence of
snapshots~\cite{araujo2014com2,dunlavy2011temporal,tantipathananandh2007framework}.
These techniques fail to fully capture the richness of the temporal information
in the data.

Characterizing temporal networks also brings a number of interesting challenges that
distinguish it from the analysis of static networks. For example, while the
number of nodes and pairs of connected nodes can be of manageable size, the
number of temporal edges may be very large and thus efficient algorithms are
needed when analyzing such data. Another interesting challenge is that patterns
in temporal networks can occur at different time scales.  For example, in
telephone call networks, reciprocation (that is, a person returning a call) can
occur on very short time intervals, while more intricate patterns (e.g., person
$A$ calling person $B$, who then calls $C$) may occur at larger time scales.
Lastly, there are many possible temporal
patterns as the order as well as the sequence of edges play an important role.

\xhdr{Present work: Temporal network motifs}
Here, we provide a general methodology for analyzing temporal networks. We define
temporal networks as a set of nodes and a collection of directed temporal edges, where
each edge has a timestamp.  For example,
Fig.~\ref{fig:intro1A} illustrates a small temporal network with nine temporal
edges between five ordered pairs of nodes.

Our analytical approach is based on generalizing the notion of network motifs to
temporal networks. In static networks, network motifs or graphlets are defined
as small induced subgraphs occurring in a bigger network
structure~\cite{benson2016higher,milo2002network,yaverouglu2014revealing}. We
extend static motifs to temporal networks and define $\delta$-temporal motifs,
where all the edges in a given motif $M$ have to occur inside the time period of
$\delta$ time units. These $\delta$-temporal motifs simultaneously account for
ordering of edges and a temporal window in which edges can occur.  For example,
Fig.~\ref{fig:intro1B} shows a motif on three nodes and three edges, where the
edge label denotes the order in which the edges appear.  While we focus on
directed edges with a single timestamp in this work, our methodology seamlessly
generalizes to common variations on this model.  For example, our methods can
incorporate timestamps with durations (common in telephone call networks),
colored edges that identify different types of connections, and temporal
networks with undirected edges.

We then consider the problem of counting how many times does each
$\delta$-temporal motif occur in a given temporal network.  We develop a general
algorithm for counting temporal network motifs defined by any number of nodes
and edges that avoids enumeration over subsets of temporal edges and whose
complexity depends on the structure of the static graph induced by the temporal
motif.  For motifs defined by a constant number of temporal edges between $2$
nodes, this general algorithm is optimal up to constant factors---it runs in
$O(m)$ time, where $m$ is the number of temporal edges.

Furthermore, we design fast variations of the algorithm that allow for counting
certain classes of $\delta$-temporal motifs including star and triangle
patterns.  These algorithms are based on a common framework for managing summary
counts in specified time windows.  For star motifs with $3$ nodes and $3$
temporal edges, we again achieve a running time linear in the input, i.e.,
$O(m)$ time. Given a temporal graph with $\tau$ induced triangles in its induced static graph,
our fast algorithm counts temporal triangle motifs with $3$ temporal
edges in $O(\tau^{1/2}m)$  worst-case time.  In contrast, any algorithm that
processes triangles individually takes $O(\tau m)$ worst-case time.  In
practice, our fast temporal triangle counting algorithm is up to 56 times faster
than a competitive baseline and runs in just a couple of hours on a network with
over two billion temporal edges.

Our algorithmic framework enables us to study the structure of several complex
systems.  For example, we explore the differences in human communication
patterns by analyzing motif frequencies in text message, Facebook wall post,
email and private online message network datasets.  Temporal
network motif counts reveal that text messaging and Facebook wall posting are
dominated by ``blocking'' communication, where a user only engages with one other
user at a time, whereas email is mostly characterized by ``non-blocking''
communication as individuals send out several emails in a row.  Furthermore,
private online messaging contains a mixture of blocking and non-blocking
behavior.

Temporal network motifs can also be used to measure the frequency of patterns at
different time scales.  For example, the difference in $\delta$-temporal motif
counts for $\delta = 60$ minutes and $\delta = 30$ minutes counts only the
motifs that take at least 30 minutes and at most 60 minutes to form.  With this
type of analysis, we find that certain question-and-answer patterns on Stack
Overflow need at least 30 minutes to develop.  We also see that in online
private messaging, star patterns constructed by outgoing messages sent by one
user tend to increase in frequency from time scales of 1 to 20 minutes before
peaking and then declining in frequency.

All in all, our work defines a flexible notion of motifs in temporal networks
and provides efficient algorithms for counting them. It enables new analyses in
a variety of scientific domains and paves a new way for modeling dynamic complex
systems.

\setlength{\textfloatsep}{3pt}

\section{Related work}
\label{sec:related}

Our work builds upon the rich literature on network motifs in static graphs,
where these models have proved crucial to understanding the mechanisms driving
complex systems~\cite{milo2002network} and to characterizing classes of static
networks~\cite{vazquez2004topological,yaverouglu2014revealing}.  Furthermore,
motifs are critical for understanding the higher-order organizational patterns
in networks~\cite{benson2015tensor,benson2016higher}.  On the algorithmic side,
a large amount of research has been devoted simply to counting triangles in
undirected static graphs~\cite{latapy2008main}.

Prior definitions of temporal network motifs either do not account for edge
ordering~\cite{zhao2010communication}, only have heuristic counting
algorithms~\cite{gurukar2015commit}, or assume temporal edges in a motif must be
consecutive events for a node~\cite{kovanen2011temporal}.  In the last case, the
restrictive definition permits fast counting algorithms but misses important
structures.  For example, many related edges occurring in a short burst at a node
would not be counted together.  In contrast, $\delta$-temporal motifs
capture \emph{every} occasion that edges form a particular pattern within the
prescribed time window.

There are several studies on pattern formation in growing networks where one
only considers the addition of edges to a static graph over time.  In this
context, motif-like patterns have been used to create evolution rules that
govern the ways that networks
develop~\cite{berlingerio2009mining,ugander2013subgraph}.  The way we consider
ordering of temporal edges in our definition of $\delta$-temporal motifs is
similar in spirit.  There are also several analyses on the formation of
triangles in a variety of social
networks~\cite{huang2014mining,kossinets2006empirical,leskovec2010signed}.  In
contrast, in the temporal graphs we study here, three nodes may form a triangle
several times.

\section{Preliminaries}
\label{sec:preliminaries}

We now provide formal definitions of temporal graphs and $\delta$-temporal
motifs.  In Section~\ref{sec:algorithms}, we provide algorithms for counting the
number of $\delta$-temporal motifs in a given temporal graph.

\xhdr{Temporal edges and graphs}
We define a \emph{temporal edge} to be a timestamped directed edge between an
ordered pair of nodes.  We call a collection of temporal edges a \emph{temporal graph}
(Fig.~\ref{fig:intro1A}).  Formally, a temporal graph $T$ on a node set
$V$ is a collection of tuples $(u_i, v_i, t_i)$, $i = 1, \ldots, m$, where each
$u_i$ and $v_i$ are elements of $V$ and each $t_i$ is a timestamp in
$\mathbb{R}$.  We refer to a specific $(u_i, v_i, t_i)$ tuple as
a \emph{temporal edge}. There can be many temporal edges directed from $u$ to
$v$, and we refer to them as \emph{edges between $u$ and $v$}.  We assume that
the timestamps $t_i$ are unique so that the tuples may be strictly ordered. This
assumption makes the presentation of the definitions and algorithms clearer, but
our methods can easily be adapted to the case when timestamps are not unique.
When it is clear from context, we refer to a temporal edge as simply
an \emph{edge}. Finally, by ignoring timestamps and duplicate
edges, the temporal graph induces a standard directed graph, which we call
the \emph{static graph} $G$ of $T$ with \emph{static edges}, i.e.,
$(u, v)$ is an edge in $G$ if and only if there is some temporal edge $(u, v, t)$ in $T$.

\xhdr{$\boldsymbol{\delta}$-temporal motifs and motif instances}
We formalize $\delta$-temporal motifs with the following definition.
\begin{definition*}
A \emph{$k$-node, $l$-edge, $\delta$-temporal motif} is a sequence of $l$ edges,
$M = (u_1, v_1, t_1), (u_2, v_2, t_2) \ldots, (u_l, v_l, t_l)$ that are time-ordered
within a $\delta$ duration, i.e.,
$t_1 < t_2 \ldots < t_l$ and $t_l - t_1 \le \delta$,
such that the induced static graph from the edges is connected and has $k$ nodes.
\end{definition*}
Note that with this definition, many edges between the
same pair of nodes may occur in the motif $M$.  Also, we note that
the purpose of the timestamps is to induce an ordering on the edges.
Fig.~\ref{fig:intro1B} illustrates a particular $3$-node, $3$-edge
$\delta$-temporal motif.

The above definition provides a template for a particular pattern, and we are
interested in how many times a given pattern occurs in a dataset.  Intuitively,
a collection of edges in a given temporal graph is an instance of a
$\delta$-temporal motif $M$ if it matches the same edge pattern and all of the
edges occur in the right order within the $\delta$ time window
(Fig.~\ref{fig:intro1C}).  Formally, we say that any time-ordered sequence
$S = (w_1, x_1, t'_1), \ldots, (w_l, x_l, t'_l)$ of $l$ unique edges is
an \emph{instance} of the motif $M = (u_1, v_1, t_1), \ldots, (u_l, v_l, t_l)$
if
\begin{enumerate}[noitemsep, topsep=\enumtopsep]
\item There exists a bijection $f$ on the vertices such that $f(w_i) = u_i$ and $f(x_i) = v_i$, $i = 1, \ldots, l$, and
\item the edges all occur within $\delta$ time, i.e., $t'_l - t'_1 \le \delta$
\end{enumerate}

A central goal of this work is to count the number of ordered subsets of edges
from a temporal graph $T$ that are instances of a particular motif.  In other
words, given a $k$-node, $l$-edge $\delta$-temporal motif, we seek to find how
many of the $l! {m \choose l}$ ordered length-$l$ sequences of edges in the
temporal graph $T$ are instances of the motif.  A naive approach to this problem
would be to simply enumerate all ordered subsets and then check if it is an
instance of the motif.  In modern datasets, the number of edges $m$ is typically
quite large (we analyze a dataset in Section~\ref{sec:experiments} with over two
billion edges), and this approach is impractical even for $l = 2$.  In the
following section, we discuss several faster algorithms for counting the number
of instances of $\delta$-temporal motifs in a temporal graph.

\section{Algorithms}
\label{sec:algorithms}

\begin{figure}[tb]
\centering
\includegraphics[width=\columnwidth]{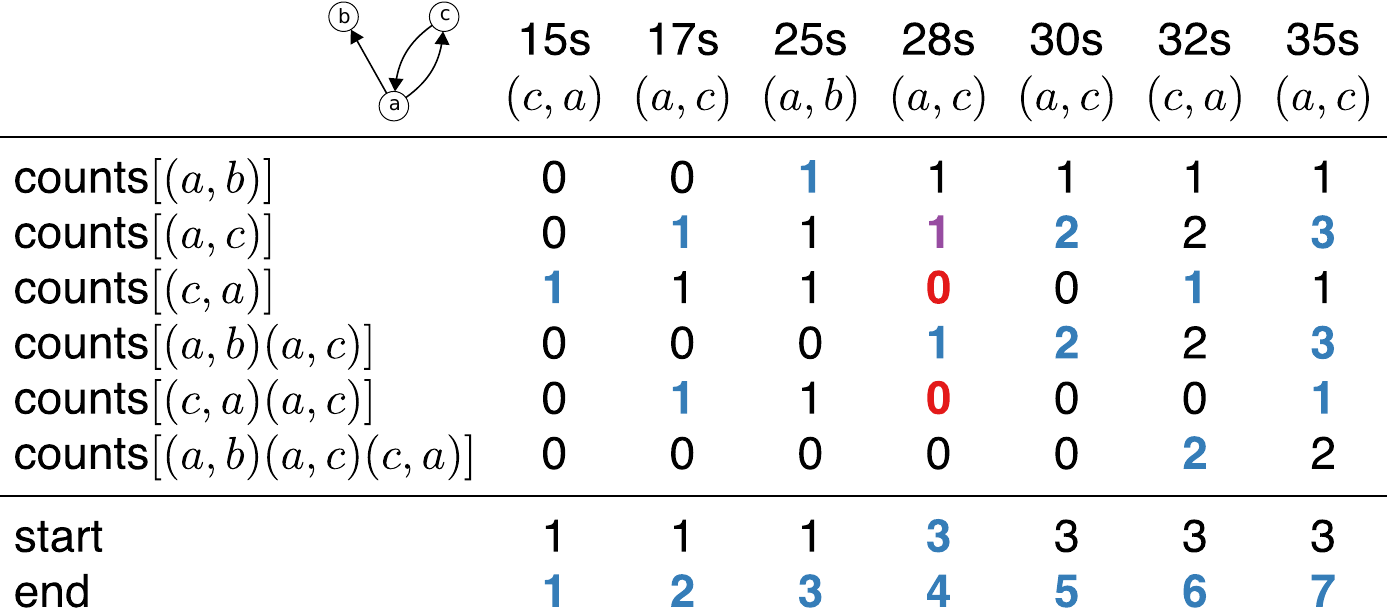}
\vspace{-1mm}
\caption{
Example execution of Alg.~\ref{alg:general} for counting instances of the
$\delta$-temporal motif $M$ in Fig.~\ref{fig:intro1}.  Each column shows the
value of counters at the end of the for loop that processes temporal edges.
Color indicates change in the variable: incremented (blue), decremented (red),
incremented and decremented (purple), or no change (black).  At the end of
execution, $\counts[(a,b)(a,c)(c,a)] = 2$ for the two instances of the temporal
motif $M$ with center node $a$.  Here we only show the counters needed to count
$M$; in total, Alg.~\ref{alg:general} maintains 39 total counters for this input
edge sequence, 25 of which are non-zero.
}
\label{fig:intro2}
\end{figure}

We now present several algorithms for exactly counting the number of instances
of $\delta$-temporal motifs in a temporal graph.  We first present a general
counting algorithm in Section~\ref{sec:general_framework}, which can count
instances of any $k$-node, $l$-edge temporal motif faster than simply
enumerating over all size-$l$ ordered subsets of edges.  This algorithm is
optimal for counting $2$-node temporal motifs in the sense that it is linear in
the number of edges in the temporal graph.  In Section~\ref{sec:faster_algs}, we
provide faster, specialized algorithms for counting specific types of $3$-node,
$3$-edge temporal motifs (Fig.~\ref{fig:three_edge_motifs}).

\subsection{General counting framework}\label{sec:general_framework}

\begin{algorithm}[tb]
  \SetKw{KwTo}{in}\SetKwFor{For}{for}{\string:}{}%
  \SetFuncSty{emph}
  \SetKwProg{myproc}{Procedure}{}{}
  \DontPrintSemicolon
  \SetKw{KwTo}{in}\SetKwFor{For}{for}{\string:}{}%
  \SetFuncSty{emph}
  \SetKwFunction{increment}{IncrementCounts}
  \SetKwFunction{decrement}{DecrementCounts}  
  \KwIn{Sequence $S'$ of edges $(e_1 = (u_1, v_1), t_1), \ldots,$ $(e_L, t_L)$ with
  $t_1 < \ldots < t_L$, time window $\delta$
  }
  \KwOut{Number of instances of each $l$-edge $\delta$-temporal motif $M$ contained in the sequence}
  \caption{Algorithm for counting the
    number of instances of all possible $l$-edge $\delta$-temporal motifs
    in an ordered sequence of temporal edges.
    We assume the keys of $\counts[\cdot]$ are accessed in order of length.}
  $\tstart \leftarrow 1$,  $\counts \leftarrow$ Counter(default = 0)\;
  \For{$\textnormal{end} = 1, \ldots, L$}{
  	\While{$t_{\tstart} + \delta < t_{\tend}$}{
            \decrement{$e_{\tstart}$},
	        $\tstart \pluseq 1$\;
	}
	\increment{$e_{\tend}$}\;
  }
  \Return{$\counts$}\;
  \myproc{\decrement{$e$}}{
    $\counts[e] \minuseq 1$\;
      \For{\textnormal{$\suffix$ \KwTo $\counts.\keys$ of length $< l - 1$}}{
          $\counts[\concat(e, \suffix)] \minuseq \counts[\suffix]$\;
    }
 }
  \myproc{\increment{$e$}}{
      \For{$\prefix$ \KwTo $\counts.\keys.\reverse()$ of length $< l $}{
            $\counts[\concat(\prefix, e)] \pluseq \counts[\prefix]$
        }
      $\counts[e] \pluseq 1$\;
   }
  \label{alg:general}
\end{algorithm}

We begin with a general framework for counting the number of instances of a
$k$-node, $l$-edge temporal motif $M$.  To start, consider $H$ to be the static
directed graph induced by the edges of $M$. A sequence of temporal edges $S$ is
an instance of $M$ if and only if the static subgraph induced by edges in $S$ is
isomorphic to $H$, the ordering of the edges in $S$ matches the order in $M$,
and all the edges in $S$ span a time window of at most $\delta$ time units.
This leads to the following general algorithm for counting instances of $M$ in a
temporal graph $T$:
\begin{enumerate}[noitemsep,topsep=10pt]
\item
Identify all instances $H'$ of the static motif $H$ induced by $M$ within the
static graph $G$ induced by the temporal graph $T$ (e.g., there are three 
instances of $H$ induced by $M$ in Fig.~\ref{fig:intro1}).
\item
For each static motif instance $H'$, gather all temporal edges between pairs of
nodes forming an edge in $H'$ into an ordered sequence $S' =$ $(u_1, v_1, t_1)$,
$\ldots$, $(u_L, v_L, t_L)$.
\item
Count the number of (potentially non-contiguous) subsequences of edges in $S'$
occurring within $\delta$ time units that correspond to instances of $M$.
\end{enumerate}

The first step can use known algorithms for enumerating motifs in static
graphs~\cite{wernicke2006fanmod}, and the second step is a simple matter of
fetching the appropriate temporal edges.  To perform the third step efficiently,
we develop a dynamic programming approach for counting the number of
subsequences (instances of motif $M$) that match a particular pattern within a
larger sequence ($S'$).  The key idea is that, as we stream through an input
sequence of edges, the count of a given length-$l$ pattern (i.e., motif) with a
given final edge is computed from the current count of the length-($l-1$) prefix
of the pattern.  Inductively, we maintain auxiliary counters of all of the
prefixes of the pattern (motif).
Second, we also require that all edges in the motif be at most $\delta$ time apart.
Thus, we use the notion of a moving time window such that any two edges in the
time window are at most $\delta$ time apart.  The auxiliary counters now keep
track of only the subsequences occurring within the current time window.
Last, it is important to note that the algorithm only \emph{counts} the number of instances of motifs rather than \emph{enumerating} them. 

Alg.~\ref{alg:general} counts \emph{all} possible $l$-edge motifs that occur in
a given sequence of edges.  The data structure $\counts[\cdot]$ maintains
auxiliary counts of all (ordered) patterns of length at most $l$. Specifically,
$\counts[e_1\cdots e_r]$ is the number of times the subsequence $[e_1\cdots
  e_r]$ occurs in the current time window (if $r < l$) or the number of times
the subsequence has occurred within all time windows of length $\delta$ (if $r =
l$).  We also assume the keys of $\counts[\cdot]$ are accessed in order of
length.  Moving the time window forward by adding a new edge into the window,
all edges $(e = (u, v), t)$ farther than $\delta$ time from the new edge are
removed from the window and the appropriate counts are decremented (the
\emph{DecrementCounts()} method).  First, the single edge counts ($[e]$) are
updated. Based on these updates, length-$2$ subsequences formed with $e$ as its
first edge are updated and so on, up through length-($l-1$) subsequences. On the
other hand, when an edge $e$ is added to the window, similar updates take place,
but in reverse order, from longest to shortest subsequences, in order to
increment counts in subsequences where $e$ is the last edge (the
\emph{IncrementCounts()} method). Importantly, length-$l$ subsequence counts are
incremented in this step but never decremented.  As the time window moves from
the beginning to the end of the sequence of edges, the algorithm
accumulates counts of all length-$l$ subsequences in all possible time windows
of length $\delta$.

Fig.~\ref{fig:intro2} shows the execution of the Alg.~\ref{alg:general} for a
particular sequence of edges.  Note that the figure only displays values of
$\counts[\cdot]$ for contiguous subsequences of the motif $M$, but the algorithm
keeps counts for other subsequences as well.  In general, there are $O(l^2)$
contiguous subsequences of an $l$-edge motif $M$, and there are $O(\lvert H
\rvert^l)$ total keys in $\counts[\cdot]$, where $\lvert H \rvert$ is the number
of edges in the static subgraph $H$ induced by $M$, in order to count all
$l$-edge motifs in the sequence (i.e., not just motif $M$).

We now analyze the complexity of the overall 3-step algorithm.  We assume that
the temporal graph $T$ has edges sorted by timestamps, which is reasonable if
edges are logged in their order of occurrence, and we pre-process $T$ in linear
time such that we can access the sorted list of all edges between $u$ and $v$ in
$O(1)$ time.  Constructing the time-sorted sequence $S'$ in step 2 of the
algorithm then takes $O(\log({\lvert H \rvert}) {\lvert S' \rvert})$ time.
Each edge inputted to
Alg.~\ref{alg:general} is processed exactly twice: once to increment counts when
it enters the time window and once to decrement counts when it exits the time
window.  As presented in Alg.~\ref{alg:general}, each update changes $O(\vert H \vert^l)$
counters resulting in an overall complexity of $O(\vert H \vert^l {\lvert S' \rvert})$.
However, one could modify Alg.~\ref{alg:general} to only update counts for
contiguous subsequences of the sequence $M$, which would change $O(l^2)$
counters and have overall complexity $O(l^2 {\lvert S' \rvert})$.  We are typically only
interested in small constant values of $\vert H \rvert$ and $l$ (for our
experiments in Section~\ref{sec:experiments}, $\lvert H \rvert \le 3$ and $l =
3$), in which case the running time is linear in the size of the input to the
algorithm, i.e., $O({\lvert S' \rvert})$.

In the remainder of this section we analyze our 3-step algorithm with respect to
different types of motifs (2-node, stars, and triangles) and argue benefits as
well as deficiencies of the proposed framework.  We show that for $2$-node
motifs, our general counting framework takes time linear in the total number of
edges $m$.  Since all the input data needs to be examined for computing exact
counts, this means the algorithm is optimal for $2$-node motifs. However, we
also show that for star and triangle motifs the algorithm is not optimal, which
then motivates us to design faster algorithms in Sec.~\ref{sec:faster_algs}.

\xhdr{General algorithm for 2-node motifs}
We first show how to map $2$-node motifs to the framework described above.  Any
induced graph $H$ of a $2$-node $\delta$-temporal motif is either a single or a
bidirectional edge.  In either case, it is straightforward to enumerate over all
instances of $H$ in the static graph.  This leads to the following procedure:
(1) for each pair of nodes $u$ and $v$ for which there is at least one edge,
gather and sort the edges in either direction between $u$ and $v$; (2) call
Alg.~\ref{alg:general} with these edges.  The obtain the total motif count
the counts from each call to
Alg.~\ref{alg:general} are then summed together.

We only need to input each edge to Alg.~\ref{alg:general} once, and
under the assumption that we can
access the sorted directed edges from one node to another in $O(1)$ time, the
merging of edges into sorted order takes linear time.  Therefore, the total
running time is $O(2^lm)$, which is linear in the number of temporal edges $m$.  
We are mostly interested in small patterns, i.e., cases
when $l$ is a small constant.  Thus, this methodology is optimal (linear in the
input, $m$) for counting $2$-node $\delta$-temporal motif instances.

\begin{figure}[tb]
\centering
\resizebox {\columnwidth} {!} {\input{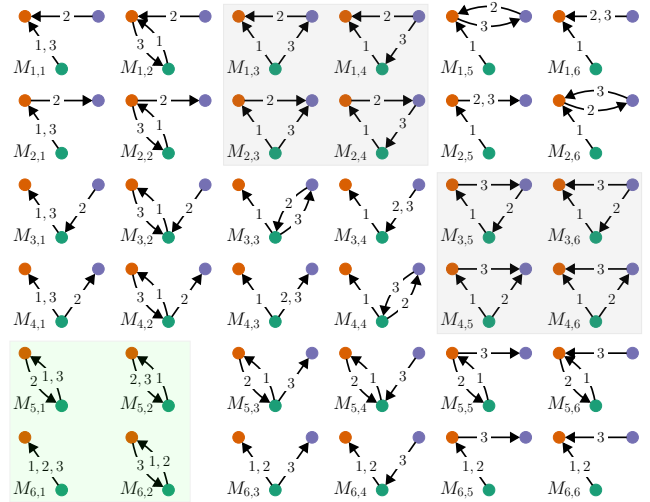}}
\caption{%
All $2$-node and $3$-node, $3$-edge $\delta$-temporal motifs.  The green
background highlights the four $2$-node motifs (bottom left) and the grey
background highlights the eight triangles.  The 24 other motifs are stars.  We
index the 36 motifs $M_{i,j}$ by 6 rows and 6 columns.  The first edge in each
motif is from the green to the orange node.  The second edge is the same along
each row, and the third edge is the same along each column. 
}
\label{fig:three_edge_motifs}
\end{figure}

\xhdr{General algorithm for star motifs}
Next, we consider $k$-node, $l$-edge star motifs $M$, whose induced static graph
$H$ consists of a center node and $k-1$ neighbors, where edges may occur in
either direction between the center node and a neighbor node.  For example, in the
top left corner of Fig.~\ref{fig:three_edge_motifs}, $M_{1,1}$ is a star motif
with all edges pointing toward the center node.  In such motifs, the induced
static graph $H$ contains at most $2(k - 1) = 2k - 2$ static edges---one
incoming and outgoing edge from the center node to each neighbor node.  We have the
following method for counting the number of instances of $k$-node, $l$-edge star
motifs: (1) for each node $u$ in the static graph and for each unique set of $k
- 1$ neighbors, gather and sort the edges in either direction between $u$ and
the neighbors; (2) count the number of instances of $M$ using
Alg.~\ref{alg:general}.  The counts from each call to
Alg.~\ref{alg:general} are summed over all center nodes.

The major drawback of this approach is that we have to loop over each
size-$(k-1)$ neighbor set.  This can be prohibitively expensive even when $k = 3$
if the center node has large degree.  In Section~\ref{sec:faster_algs}, we
shall design an algorithm that avoids this issue for the case when the star
motif has $l = 3$ edges and $k = 3$.

\xhdr{General algorithm for triangle motifs}
In triangle motifs, the induced graph $H$ consists of 3 nodes and at least one
directed edge between any pair of nodes (see Fig.~\ref{fig:three_edge_motifs}
for all eight of the $3$-edge triangle motifs).  The induced static graph $H$ of $M$
contains at least three and at most six static edges.  A straightforward
algorithm for counting $l$-edge triangle motifs in a temporal graph $T$ is:
\begin{enumerate}[noitemsep,topsep=\enumtopsep]
\item Use a fast static graph triangle enumeration algorithm to find all
  triangles in the static graph $G$ induced by $T$~\cite{latapy2008main}.
\item For each triangle $(u, v, w)$, merge all temporal edges from each pair of nodes to
  get a time-sorted list of edges.  Use Alg.~\ref{alg:general} to count the
  number of instances of $M$.
\end{enumerate}
This approach is problematic as the edges between a pair of nodes may
participate in many triangles. Fig.~\ref{fig:worstcase} shows a worst-case example for
the motif $M =$ $(w, u, t_1')$, $(w, v, t_2')$, $(u, v, t_3')$ with $\delta = \infty$.
In this case, the timestamps are ordered by their index. There are $m-2n$ edges between $u$ and
$v$, and each of these edges forms an instance of $M$ with every $w_i$. Thus,
the overall worst-case running time of the algorithm is $O(\triangleenum +
m\tau)$, where $\triangleenum$ is the time to enumerate the number of triangles
$\tau$ in the static graph.  In the following section, we devise an algorithm
that significantly reduces the dependency on $\tau$ from linear to sub-linear
(specifically, $\sqrt{\tau}$) when there are $l = 3$ edges.

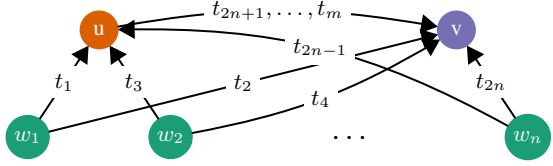
\begin{figure}[t]
\begin{center}
\scalebox{0.95}{

\definecolor{mygreen}{RGB}{27,158,119}
\definecolor{myorange}{RGB}{217,95,2}
\definecolor{mypurple}{RGB}{117,112,179}

\renewcommand*{\VertexInterMinSize}{15pt}
\renewcommand*{\VertexSmallMinSize}{15pt}

\begin{tikzpicture}
\renewcommand*{\VertexLineColor}{myorange}
\renewcommand*{\VertexLightFillColor}{myorange}
\renewcommand*{\VertexTextColor}{white}
\Vertex[x=3.5,y=3.5]{u}
\renewcommand*{\VertexLineColor}{mypurple}
\renewcommand*{\VertexLightFillColor}{mypurple}
\Vertex[x=8.5,y=3.5]{v}

\node[draw=none] (vd) at (7,2) {\large $\ldots$};

\renewcommand*{\VertexLineColor}{mygreen}
\renewcommand*{\VertexLightFillColor}{mygreen}
\node[VertexStyle](w1) at (2.5,2) {$w_1$};
\node[VertexStyle](w2) at (4.5,2) {$w_2$};
\node[VertexStyle](wn) at (9.5,2) {$w_n$};

\SetUpEdge[style={->,>=triangle 60, black, bend left=10},
  label={\footnotesize $t_{2n+1}, \ldots, t_{m}$}]
\Edge(u)(v)
\SetUpEdge[style={->,>=triangle 60, black}]

\Edge[label=$t_1$](w1)(u)
\Edge[label=$t_2$](w1)(v)
\Edge[label=$t_3$](w2)(u)
\Edge[label=$t_4$, style = {->, >=triangle 60,bend right=10}](w2)(v)
\Edge[label=$t_{2n-1}$, style={->,>=triangle 60,bend right=15}](wn)(u)
\Edge[label=$t_{2n}$](wn)(v)

\node[draw=none] (vd) at (7,2) {\large $\ldots$};

\end{tikzpicture}}
\end{center}
\caption{Worst-case example for counting triangular motifs with
Alg.~\ref{alg:general}.}
\label{fig:worstcase}
\end{figure}

\subsection{Faster algorithms}\label{sec:faster_algs}

The general counting algorithm from the previous subsection counts the number of
instances of any $k$-node, $l$-edge $\delta$-temporal motif, and is also optimal
for $2$-node motifs.  However, the computational cost may be expensive for other
motifs such as stars and triangles.  We now develop specialized algorithms that
count certain motif classes faster.  Specifically, we design faster algorithms
for counting all 3-node, 3-edge star and triangle motifs
(Fig.~\ref{fig:three_edge_motifs} illustrates these motifs).  Our algorithm for
stars is linear in the input size, so it is optimal up to constant factors.

\xhdr{Fast algorithm for 3-node, 3-edge stars}
With $3$-node, $3$-edge star motifs, the key drawback of using the previous algorithmic approach would be that
we would have to loop over all pairs of neighbors given a center node.  Instead, we
will count all instances of star motifs for a given center node in just a single
pass over the edges adjacent to the center node.

We use a dynamic programming approach for counting star motifs. 
First, note that every temporal edge in a star with center $u$
is defined by (1) a neighbor node, (2) a direction of the edge (outward from or
inward to $u$), and (3) the timestamp. With this insight we then notice that there are 3 classes of star motifs on 3 nodes and 3 edges:
\begin{center}
\vspace{-0.1cm}
\scalebox{0.6}{

\begin{tikzpicture}

\colorlet{TufteRed}{red!80!black}
\newcommand{\ax}{0.0}
\newcommand{\ay}{0.0}
\newcommand{\bx}{-1.25}
\newcommand{\by}{1.25}
\newcommand{\cx}{1.25}
\newcommand{\cy}{1.25}
\newcommand{\tfirst}{1}
\newcommand{\tsecond}{2}
\newcommand{\tthird}{3}
\newcommand{\figfontsize}{\Large}
\newcommand{\xshift}{4}

\renewcommand*{\VertexLineColor}{white}
\node[draw=none] (pre) at (\ax,\by) {\Large pre};
\node[draw=none] (pre) at (\xshift+\ax,\by) {\Large post};
\node[draw=none] (pre) at (2*\xshift+\ax,\by) {\Large mid};

\renewcommand*{\VertexInterMinSize}{11pt}
\renewcommand*{\VertexSmallMinSize}{11pt}
\definecolor{mygreen}{RGB}{27,158,119}
\definecolor{myorange}{RGB}{217,95,2}
\definecolor{mypurple}{RGB}{117,112,179}

\SetVertexNoLabel
\renewcommand*{\VertexLineColor}{mygreen}
\renewcommand*{\VertexLightFillColor}{mygreen}
\Vertex[x=\ax,y=\ay]{a}
\renewcommand*{\VertexLineColor}{myorange}
\renewcommand*{\VertexLightFillColor}{myorange}
\Vertex[x=\bx,y=\by]{b}
\renewcommand*{\VertexLineColor}{mypurple}
\renewcommand*{\VertexLightFillColor}{mypurple}
\Vertex[x=\cx,y=\cy]{c}
\SetUpEdge[style={-,ultra thick,black}, label={\figfontsize $\tfirst,\tsecond$}]\Edge(a)(b)
\SetUpEdge[style={-,ultra thick,black}, label={\figfontsize $\tthird$}]\Edge(a)(c)

\renewcommand*{\VertexLineColor}{mygreen}
\renewcommand*{\VertexLightFillColor}{mygreen}
\Vertex[x=\xshift+\ax,y=\ay]{a}
\renewcommand*{\VertexLineColor}{myorange}
\renewcommand*{\VertexLightFillColor}{myorange}
\Vertex[x=\xshift+\bx,y=\by]{b}
\renewcommand*{\VertexLineColor}{mypurple}
\renewcommand*{\VertexLightFillColor}{mypurple}
\Vertex[x=\xshift+\cx,y=\cy]{c}
\SetUpEdge[style={-,ultra thick,black}, label={\figfontsize $\tfirst$}]\Edge(a)(b)
\SetUpEdge[style={-,ultra thick,black}, label={\figfontsize $\tsecond,\tthird$}]\Edge(a)(c)

\renewcommand*{\VertexLineColor}{mygreen}
\renewcommand*{\VertexLightFillColor}{mygreen}
\Vertex[x=2*\xshift+\ax,y=\ay]{a}
\renewcommand*{\VertexLineColor}{myorange}
\renewcommand*{\VertexLightFillColor}{myorange}
\Vertex[x=2*\xshift+\bx,y=\by]{b}
\renewcommand*{\VertexLineColor}{mypurple}
\renewcommand*{\VertexLightFillColor}{mypurple}
\Vertex[x=2*\xshift+\cx,y=\cy]{c}
\SetUpEdge[style={-,ultra thick,black}, label={\figfontsize $\tfirst,\tthird$}]\Edge(a)(b)
\SetUpEdge[style={-,ultra thick,black}, label={\figfontsize $\tsecond$}]\Edge(a)(c)

\end{tikzpicture}
}
\vspace{-0.1cm}  
\end{center}
where each class has $2^3 = 8$ motifs for each of the possible directions on the
three edges.

Now, suppose we process the time-ordered
sequence of edges containing the center node $u$.  We maintain the
following counters when processing an edge with timestamp $t_j$:
\begin{itemize}[noitemsep,topsep=\enumtopsep]
\item $\presum[\dir1, \dir2]$ is the number of sequentially ordered pairs of
  edges in $[t_j - \delta, t_j)$ where the first edge points in direction
    $\dir1$ and the second edge points in direction $\dir2$
\item $\postsum[\dir1, \dir2]$ is the analogous counter for the time window
  $(t_j, t _j + \delta]$.
\item $\middlesum[\dir1, \dir2]$ is the number of pairs of edges where the first
  edge is in direction $\dir1$ and occurred at time $t < t_j$ and the second edge
  is in direction $\dir2$ and occurred at time $t' > t_j$ such that $t' - t \le
  \delta$.
\end{itemize}

If we are currently processing an edge, the ``pre'' class gets
$\presum[\dir1, \dir2]$ new motif instances for any choice of directions $\dir1$
and $\dir2$ (specifying the first two edge directions) and the current edge serves as
the third edge in the motif (hence specifying the third edge direction).
Similar updates are made with the $\postsum[\cdot,\cdot]$ and $\middlesum[\cdot,\cdot]$
counters, where the current edge serves as the first or second edge in the motif, respectively.

In order for our algorithm to be efficient, we must quickly update our counters.
To aid in this, we introduce two additional counters:
\begin{itemize}[noitemsep,topsep=\enumtopsep]
\item $\prenodecount[\dir, v_i]$ is the number of times node $v_i$ has appeared
  in an edge with $u$ with direction $\dir$ in the time window $[t_j - \delta, t_j)$
\item $\postnodecount[\dir, v_i]$ is the analogous counter but for the time
  window $(t_j, t_j + \delta]$.
\end{itemize}

Following the ideas of Alg.~\ref{alg:general}, it is easy to
update these counters when we process a new edge.
Consequently, $\presum[\cdot,\cdot]$, $\postsum[\cdot,\cdot]$, and
$\middlesum[\cdot,\cdot]$ can be maintained when processing an edge
with just a few simple primitives:
\begin{itemize}[noitemsep,topsep=\enumtopsep]
\item
  \emph{Push()} and \emph{Pop()} update the counts for $\prenodecount[\cdot,\cdot]$,
  $\postnodecount[\cdot,\cdot]$, $\presum[\cdot,\cdot]$ and $\postsum[\cdot,\cdot]$ when edges enter and leave the time windows $[t_j - \delta, t_j)$ and $(t_j, t_j + \delta]$.
\item
  \emph{ProcessCurrent()} updates motif counts involving the current edge
  and updates the counter $\middlesum[\cdot,\cdot]$.
\end{itemize}

We describe the general procedure in Alg.~\ref{alg:fast_framework}, which
will also serve as the basis for our fast triangle counting procedure, and
Alg.~\ref{alg:fast_wedges} implements the subroutines \emph{Push()},
\emph{Pop()}, and \emph{ProcessCurrent()} for counting instances of $3$-node, $3$-edge star
motifs.  The $\globalpre[\cdot,\cdot,\cdot]$, $\globalpost[\cdot,\cdot,\cdot]$, and $\globalmid[\cdot,\cdot,\cdot]$ counters in
Alg.~\ref{alg:fast_wedges} maintain the counts of the three different
classes of stars described above.

Finally, we note that our counting scheme incorrectly includes instances of
$2$-node motifs such as $M = (u, v_i, t_1)$, $(u, v_i, t_2)$, $(u, v_i, t_3)$,
but we can use the efficient $2$-node motif counting algorithm to account for
this.  Putting everything together, we have the following procedure:
\begin{enumerate}[noitemsep,topsep=\enumtopsep]
\item For each node $u$ in the temporal graph $T$, get a time-ordered
  list of all edges containing $u$.
\item Use Algs.~\ref{alg:fast_framework}~and~\ref{alg:fast_wedges} to count
  star motif instances.
\item For each neighbor $v$ of a star center $u$, subtract the 2-node motif
  counts using Alg.~\ref{alg:general}.
\end{enumerate}

If the $m$ edges of $T$ are time-sorted, the first step can be done in linear
time.  The second and third steps run in linear time in the input size.  Each
edge is used in steps 2 and 3 exactly twice: once for each end point as the
center node.  Thus, the overall complexity of the algorithm is $O(m)$,
which is optimal up to constant factors.

\begin{algorithm}[tb]
  \SetKw{KwTo}{in}\SetKwFor{For}{for}{\string:}{}%
  \SetFuncSty{emph}
  \SetKwProg{myproc}{Procedure}{}{}
  \DontPrintSemicolon
  \caption{Algorithmic framework for faster counting of $3$-node, $3$-edge star and
  triangle temporal motifs.  The fast star counting method (Alg.~\ref{alg:fast_wedges})
  and triangle counting method (Alg.~\ref{alg:fast_triangles}) implement
  different versions of the \emph{Push()}, \emph{Pop()}, and \emph{ProcessCurrent()} subroutines.}
  \SetKwFunction{push}{Push}
  \SetKwFunction{pop}{Pop}
  \SetKwFunction{process}{ProcessCurrent}
  \KwIn{Sequence of edges $(e_1 = (u_1, v_1), t_1), \ldots, (e_L, t_L)$ with
  $t_1 < \ldots < t_L$, time window $\delta$}
  Initialize counters $\prenodecount$, $\postnodecount$, $\middlesum$, $\presum$, and $\postsum$;
  $\tstart \leftarrow 1$, $\tend \leftarrow 1$\;
  \For{$j = 1, \ldots, L$}{
      \While{$t_{\tstart} + \delta < t_j$}{
         \pop{$\prenodecount$, $\presum$, $e_{\tstart}$},
         $\tstart \pluseq 1$
       }
      \While{$t_{\tend} \le t_j + \delta$}{
          \push{$\postnodecount$, $\postsum$, $e_{\tend}$}, $\tend \pluseq 1$\;
       }
      \pop{$\postnodecount$, $\postsum$, $e_j$}\;
       \process{$e_j$}\;
       \push{$\prenodecount$, $\presum$, $e_j$}\;
  }
  \label{alg:fast_framework}
\end{algorithm}

\begin{algorithm}[tb]
  \SetKw{KwTo}{in}\SetKwFor{For}{for}{\string:}{}%
  \SetFuncSty{emph}
  \SetKwProg{myproc}{Procedure}{}{}
  \DontPrintSemicolon
  \caption{Implementation of Alg.~\ref{alg:fast_framework} subroutines for
    efficiently counting instances of $3$-node, $3$-edge star motifs.  Temporal
    edges are specified by a neighbor $\nbr$, a direction $\dir$ (incoming or
    outgoing), and a timestamp.  The ``:'' notation represents a selection of
    all indices in an array.}
  \SetKwFunction{push}{Push}
  \SetKwFunction{pop}{Pop}
  \SetKwFunction{process}{ProcessCurrent}
  Initialize counters $\globalpre$, $\globalpost$, $\globalmid$\;
  \myproc{\push{$\nodecount$, $\sumtxt$, $e = (\nbr, \dir)$}}{
    $\sumtxt[:, \dir] \pluseq \nodecount[:, \nbr]$\;
    $\nodecount[\dir, \nbr] \pluseq 1$\;
  }
  \myproc{\pop{$\nodecount$, $\sumtxt$, $e = (\nbr, \dir)$}}{
    $\nodecount[\dir, \nbr] \minuseq 1$\;
    $\sumtxt[\dir, :] \minuseq \nodecount[:, \nbr]$\;
  }
  \myproc{\process{$e = (\nbr, \dir)$}}{
    $\middlesum[:, \dir] \minuseq \prenodecount[:, \nbr]$\;
    $\globalpre[:,:,\dir] \pluseq \presum[:,:]$\;
    $\globalpost[\dir, : , :] \pluseq \postsum[:,:]$\;
    $\globalmid[:, \dir, :] \pluseq \middlesum[:,:]$\;
    $\middlesum[\dir, :] \pluseq \postnodecount[:, \nbr]$\;    
  }
  \Return{$\globalpre$, $\globalpost$, $\globalmid$}
  \label{alg:fast_wedges}
\end{algorithm}

\begin{algorithm}[tb]
  \SetKw{KwTo}{in}\SetKwFor{For}{for}{\string:}{}%
  \SetFuncSty{emph}
  \SetKwProg{myproc}{Procedure}{}{}
  \DontPrintSemicolon
  \caption{Implementation of Alg.~\ref{alg:fast_framework} subroutines for
    counting $3$-edge triangle motifs containing a specified pair of nodes $u$
    and $v$.  Temporal edges are specified by a neighbor $\nbr$, a direction
    $\dir$ (incoming or outgoing), an indicator ``$\uorv$'' denoting if the edge
    connects to $u$ or $v$, and a timestamp.  The ``:'' notation represents a
    selection of all indices in an array.}
  \SetKwFunction{push}{Push}
  \SetKwFunction{pop}{Pop}
  \SetKwFunction{process}{ProcessCurrent}
  Initialize counter $\globalall$\;
  \myproc{\push{$\nodecount$, $\sumtxt$, $e = (\nbr, \dir, \uorv)$}}{
    \lIf{$\nbr \in \{u, v\}$}{\Return}
      $\sumtxt[\muorv, :, \dir] \pluseq \nodecount[\muorv, :, \nbr]$\;
    $\nodecount[\uorv,\dir,\nbr] \pluseq 1$\;
  }
  \myproc{\pop{$\nodecount$, $\sumtxt$, $e = (\nbr, \dir, \uorv)$}}{
    \lIf{$\nbr \in \{u, v\}$}{\Return}
    $\nodecount[\uorv, \dir, \nbr] \minuseq 1$\;
    $\sumtxt[\uorv, \dir, :] \minuseq \nodecount[\muorv, :, \nbr]$
  }
  \myproc{\process{$e = (\nbr, \dir, \uorv)$}}{
    \uIf{$\nbr \notin \{u, v\}$}{
      $\middlesum[\muorv, :, \dir] \minuseq \prenodecount[\muorv, :, \nbr]$\;
    $\middlesum[\uorv, \dir, :] \pluseq \postnodecount[\muorv, :, \nbr]$\;   
    }\Else{
      $\utov = (\nbr==u)$ XOR $\dir$\;
      \For{$0 \le i, j, k \le 1$}{
        $\globalall[i, j, k] \pluseq \middlesum[j \text{ XOR } \utov, i, k]$\;
        \hspace{0.75cm}$+ \postsum[i \text{ XOR } \utov, j, 1 - k]$\;
        \hspace{0.75cm}$+ \presum[k \text{ XOR } \utov, 1- i, 1 - j]$\;
      }
    }
  }
  \tcc{$\globalall$ key map to Fig.~\ref{fig:three_edge_motifs}:}
  \tcc{\scriptsize{$[0, 0, 0] \mapsto M_{1,3}$, $[0, 0, 1] \mapsto M_{1, 4}$, $[0, 1, 0] \mapsto M_{2, 3}$}
  \normalsize{}}
  \tcc{\scriptsize{$[0, 1, 1] \mapsto M_{2,4}$, $[1, 0, 0] \mapsto M_{3, 5}$, $[1, 0, 1] \mapsto M_{3, 6}$}\normalsize{}}
  \tcc{\scriptsize{$[1, 1, 0] \mapsto M_{4, 5}$, $[1, 1, 1] \mapsto M_{4, 6}$}\normalsize{}}
\Return{$\globalall$}
\label{alg:fast_triangles}
\end{algorithm}

\xhdr{Fast algorithm for 3-edge triangle motifs}
While our fast star counting routine relied on counting motif instances for all
edges adjacent to a given \emph{node}, our fast triangle algorithm is based on
counting instances for all edges adjacent to a given \emph{pair of nodes}.
Specifically, given a pair of nodes $u$ and $v$ and a list of common neighbors
$w_1, \ldots, w_d$, we count the number of motif instances for triangles $(w_i,
u, v)$.  Given all of the edges between these three nodes, the counting
procedures are nearly identical to the case of stars.  We use the same general
counting method (Alg.~\ref{alg:fast_framework}), but the behavior of the
subroutines \emph{Push()}, \emph{Pop()}, and \emph{ProcessCurrent()} depends on
whether or not the edge is between $u$ and $v$.

These methods are implemented in Alg.~\ref{alg:fast_triangles}.  The input
is a list of edges adjacent to a given pair of neighbors $u$ and $v$, where each
edge consists of four pieces of information: (1) a neighbor node $\nbr$, (2) an
indicator of whether or not the node $\nbr$ connects to node $u$ or node $v$,
(3) the direction $\dir$ of the edge, and (4) the timestamp.  
The node counters ($\prenodecount[\cdot,\cdot,\cdot]$ and $\postnodecount[\cdot,\cdot,\cdot]$) in Alg.~\ref{alg:fast_triangles} have an extra dimension compared
to Alg.~\ref{alg:fast_wedges} to indicate whether the counts correspond to
edges containing node $u$ or node $v$ (denoted by ``$\uorv$'').
Similarly, the sum counters ($\presum[\cdot,\cdot,\cdot]$, $\middlesum[\cdot,\cdot,\cdot]$ and $\postsum[\cdot,\cdot,\cdot]$)  
have an extra dimension to denote if the first edge is incident on node $u$ or node $v$.

Recall that the problem with counting triangle motifs by the general framework
in Alg.~\ref{alg:general} is that a pair of nodes with many edges might
have to be counted for many triangles in the graph.  However, with
Alg.~\ref{alg:fast_triangles}, we can simultaneously count all triangles
adjacent to a given pair of nodes.  What remains is that we must assign each
triangle in the static graph to a pair of nodes.  Here, we propose to assign
each triangle to the pair of nodes in that triangle containing the largest number of edges, which
is sketched in Alg.~\ref{alg:triangle}.
Alg.~\ref{alg:triangle} aims to process as many triangles as
possible for pairs of nodes with many edges.
The following theorem says that this is faster than simply counting for
each triangle (described in Section~\ref{sec:general_framework}).  Specifically,
we reduce $O(m\tau)$ complexity to $O(m\sqrt{\tau})$.
\begin{theorem*}\label{thm:fast_triangles}
In the worse case, Alg.~\ref{alg:triangle} runs in time $O(\triangleenum +
m\sqrt{\tau})$, where $\triangleenum$ is the time to enumerate all triangles in
the static graph $G$, $m$ is the total number of temporal edges, and $\tau$ is
the number of static triangles in $G$.
\end{theorem*}
\begin{proof}\vspace{-0.25cm}
Let $\sigma_i$ be the number of edges between the $i$th pair of nodes with at
least one edge, and let $p_i \ge 1$ be the number of times that edges on this
pair are used in a call to Alg.~\ref{alg:fast_triangles} by
Alg.~\ref{alg:triangle}.  Since Alg.~\ref{alg:fast_triangles} runs in
linear time in the number of edges in its input, the total running time is on
the order of $\sum_{i}\sigma_ip_i$.

The $\sigma_i$ are fixed, and we wish to find the values of $p_i$ that maximize
the summation.  Without loss of generality, assume that the $\sigma_i$ are in
decreasing order so that the most number of edges between a pair of nodes is
$\sigma_1$.  Consequently, $p_i \le i$. 
Note that each triangle contributes to at most a constant repeat
processing of edges for a given pair of nodes.  Hence, $\sum_{i}p_i \le c\tau$ 
for some constant $c$.  The summation
$\sum_{i}\sigma_ip_i$ is maximized when $p_1 = 1$, $p_2 = 2$, and so on up to
some index $j = O(\sqrt{\tau})$ for which $\textstyle \sum_{i=1}^{j}p_j =
\sum_{i=1}^{j}i = c\tau$.  Now given that the $p_j$ are fixed and the $\sigma_i$
are ordered, the summation is maximized when $\sigma_1 = \sigma_2 = \ldots =
\sigma_j = m / j$.  In this case, $\sum_{i}\sigma_ip_i = \sum_{i=1}^{j}(m/j)(1 +
i) = O(m\sqrt{\tau})$.
\end{proof}

\begin{algorithm}[t!]
  \SetKw{KwTo}{in}\SetKwFor{For}{for}{\string:}{}%
  \SetFuncSty{emph}
  \SetKwProg{myproc}{Procedure}{}{}
  \caption{Sketch of fast algorithm for counting the number of $3$-edge $\delta$-temporal
    triangle motifs in a temporal graph $T$.}
  \DontPrintSemicolon
   Enumerate all triangles in the undirected static graph $G$ of $T$\;
   $\sigma \leftarrow $ number of temporal edges on each static edge $e$ in $G$\;
   \ForEach{static triangle $\Delta = (e_1$, $e_2$, $e_3)$ \KwTo $G$}{
       $e_{\max} = \arg\max_{e \in \Delta}\{\sigma_{e}\}$\;
       \lForEach{$e$ \KwTo $\Delta$}{Add $e_{\max}$ to edge set $\ell_{e}$}
       \small{\tcc{$e' \in \ell_{e}$ if $e \in \Delta$ and $\Delta$ assigned to $e'$}}       
   }
   \ForEach{temporal edge $(e=(u,v), t)$ \KwTo time-sorted $T$}{
       \lForEach{$e'$ \KwTo $\ell_{e}$}{Append $(e, t)$ to temporal-edge list $a_{e'}$}
       \small{\tcc{$(e, t) \in a_{e'}$ if $e \in \Delta$ and $\Delta$ assigned to $e'$}}       
   }
   \ForEach{undirected edge $e$ \KwTo $G$}{
       Update counts using Alg.~\ref{alg:fast_triangles} with input $a_{e}$\;
   }
  \label{alg:triangle}
\end{algorithm}

\section{Experiments}
\label{sec:experiments}

\begin{figure*}[t]
\centering
\includegraphics[width=1.96\columnwidth]{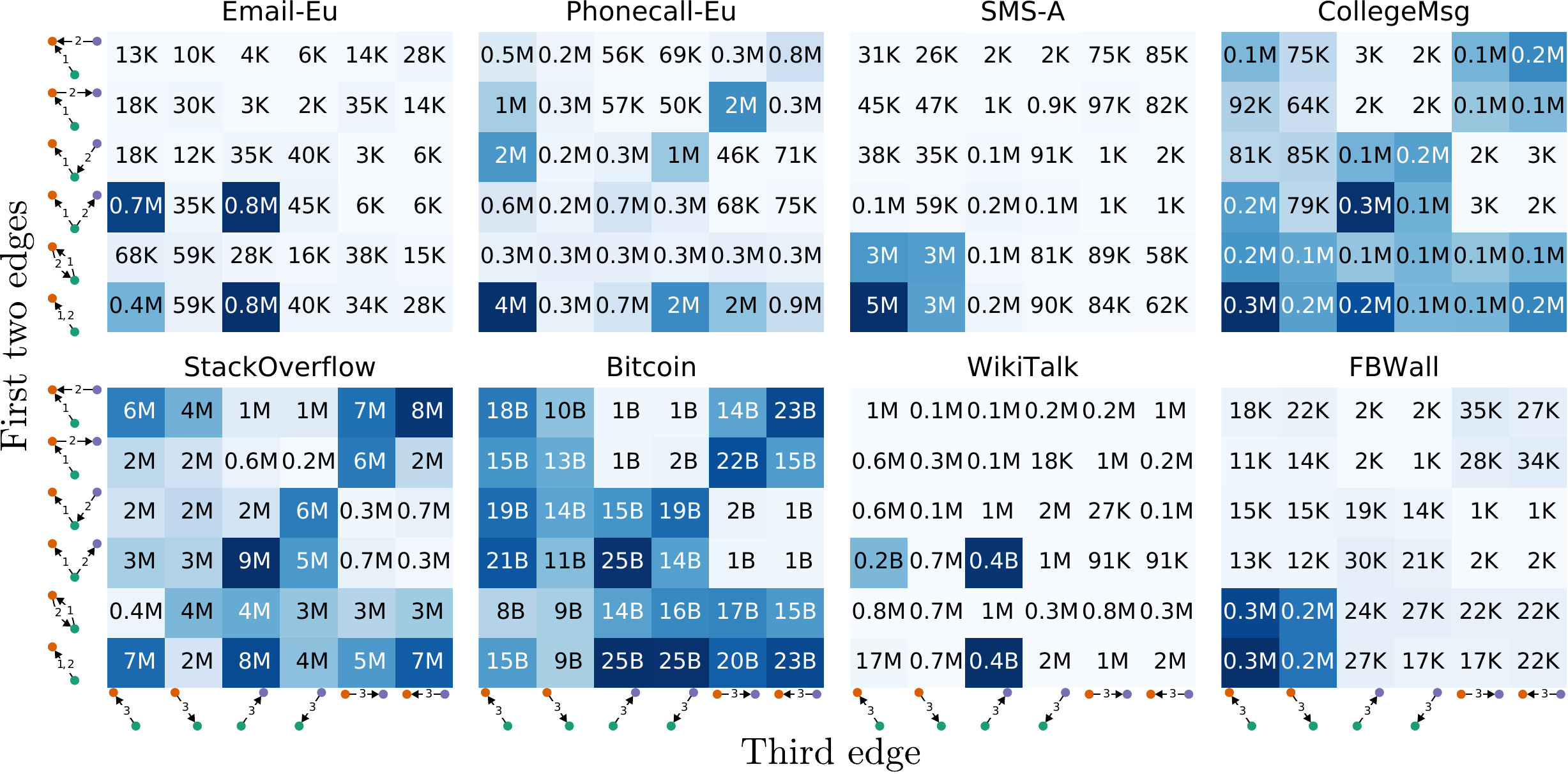}
\caption{%
Counts of instances of all $2$- and $3$-node, $3$-edge $\delta$-temporal motifs
with $\delta =$ 1 hour.  For each dataset, counts in the $i$th row and $j$th
column is the number of instances of motif $M_{i, j}$ (see
Fig.~\ref{fig:three_edge_motifs}); this motif is the union of the two edges in
the row label and the edge in the column label.  For example, there are 0.7
million instances of motif $M_{4,1}$ in the $\emaileu$ dataset. The color for
the count of motif $M_{i,j}$ indicates the fraction over all $M_{i,j}$ on a
linear scale---darker blue means a higher count.
}
\label{fig:raw_counts}
\end{figure*}

Next, we use our algorithms to reveal patterns in a variety of temporal network
datasets.  We find that the number of instances of various $\delta$-temporal
motifs reveal basic mechanisms of the networks.
Datasets and implementations of our algorithms are
available at \url{http://snap.stanford.edu/temporal-motifs}.

\subsection{Data}\label{sec:data}

\setlength{\tabcolsep}{3pt}
\begin{table}[tb]
  \centering \caption{Summary statistics of datasets.}
  \begin{tabular}{l l l l c}
    \toprule 
    dataset & \# nodes & \# static & \# edges & time span  \\
                &                & edges    &                & (days) \\
    \midrule
    \emaileu & 986 & 2.49K & 332K & 803 \\
    \phone & 1.05M & 2.74M & 8.55M & 7 \\
    \sms & 44.1K & 67.2K & 545K & 338 \\
    \messages & 1.90K & 20.3K & 59.8K & 193 \\
    \stackoverflow & 2.58M & 34.9M & 47.9M & 2774 \\
    \bitcoin  & 24.6M & 88.9M & 123M & 1811 \\
    \fbwall & 45.8K & 264K & 856K & 1560 \\
    \wikitalk & 1.09M & 3.13M & 6.10M & 2277 \\
    \phonecallme & 18.7M & 360M & 2.04B & 364 \\
    \smsme & 6.94M & 51.5M & 800M & 89 \\
    \bottomrule
    \end{tabular}\label{tab:datasets}
\end{table}  

We gathered a variety of datasets in order to study the patterns of
$\delta$-temporal motifs in several domains.  The datasets are
described below and summary statistics are in Table~\ref{tab:datasets}.  The
time resolution of the edges in all datasets is one second.

\xhdr{\emaileu}
This dataset is a collection of emails between members of a European
research institution~\cite{leskovec2007graph}.  An edge $(u, v, t)$ signifies
that person $u$ sent person $v$ an email at time $t$.

\xhdr{\phone}
This dataset was constructed from telephone call records for a major European
service provider.  An edge $(u, v, t)$ signifies that person $u$ called person
$v$ starting at time $t$.

\xhdr{\sms}
Short messaging service (SMS) is a texting service provided on mobile
phones.  In this dataset, an edge $(u, v, t)$ means that person $u$ sent
an SMS message to person $v$ at time $t$~\cite{wu2010evidence}.

\xhdr{\messages}
This dataset is comprised of private messages sent on an online social network
at the University of California, Irvine~\cite{panzarasa2009patterns}.  Users
could search the network for others and then initiate conversation based on
profile information.  An edge $(u, v, t)$ means that user $u$ sent a private
message to user $v$ at time $t$.

\xhdr{\stackoverflow}
On stack exchange web sites, users post questions and receive answers from other
users, and users may comment on both questions and answers.  We derive a
temporal network by creating an edge $(u, v, t)$ if, at time $t$, user $u$: (1)
posts an answer to user $v$'s question, (2) comments on user $v$'s question, or
(3) comments on user $v$'s answer.  We formed the temporal network from the
entirety of Stack Overflow's history up to March 6, 2016.

\xhdr{\bitcoin}
Bitcoin is a decentralized digital currency and payment system.  This dataset
consists of all payments made up to October 19, 2014~\cite{kondor2014rich}.
Nodes in the network correspond to Bitcoin addresses, and an individual may have
several addresses.  An edge $(u, v, t)$ signifies that bitcoin was transferred
from address $u$ to address $v$ at time $t$.

%

\xhdr{\fbwall}
The edges of this dataset are wall posts between users on the social network
Facebook located in the New Orleans region~\cite{viswanath2009evolution}.  Any friend of
a given user can see all posts on that user's wall, so communication is public
among friends.  An edge $(u, v, t)$ means that user $u$ posted on user $v$'s
wall at time $t$.

\xhdr{\wikitalk}
This dataset represents edits on user talk pages on
Wikipedia~\cite{leskovec2010governance}.  An edge $(u, v, t)$ signifies that
user $u$ edited user $v$'s talk page at time $t$.

\xhdr{\phonecallme\ and \smsme}
This dataset is constructed from phone call and SMS records of a large
telecommunications service provider in the Middle East.  An edge $(u, v, t)$ in
$\phonecallme$ means that user $u$ initiated a call to user $v$ at time $t$.  An edge $(u, v,
t)$ in $\smsme$ means that user $u$ sent an SMS message to user $v$ at time $t$.
We use these networks for scalability experiments in
Section~\ref{sec:scalability}.

\subsection{Empirical observations of motif counts}\label{sec:empirical}
We first examine the distribution of 2- and 3-node, 3-edge motif instance counts
from 8 of the datasets described in Section~\ref{sec:data} with $\delta = 1$
hour (Fig.~\ref{fig:raw_counts}).  We choose 1 hour for the time
window as this is close to the median time for a node to take part in three
edges in most of our datasets. We make a few empirical observations uniquely
available due to temporal motifs and provide possible explanations for
these observations.

\begin{figure}[tb]
\centering
\includegraphics[width=1.0\columnwidth]{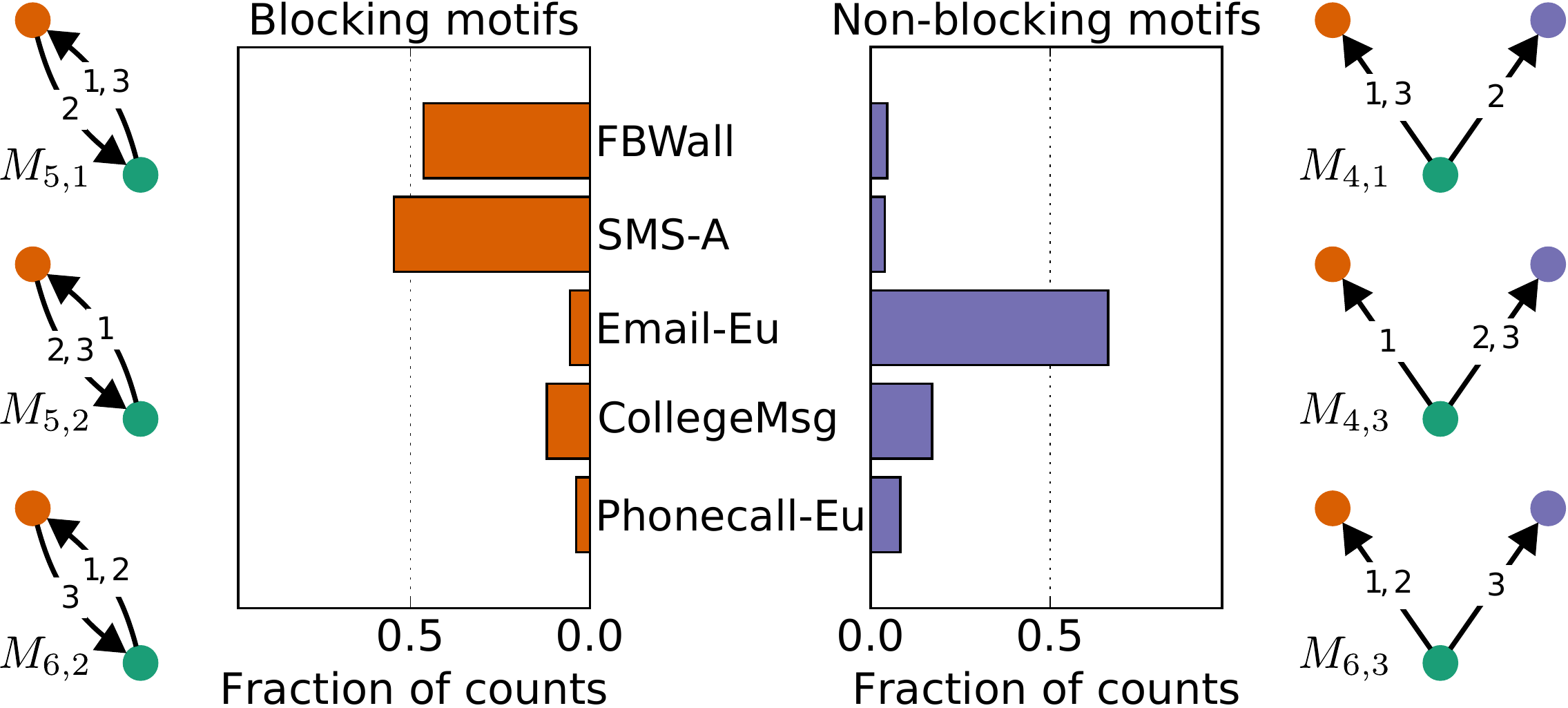}
\caption{%
Fraction of all $2$ and $3$-node, $3$-edge $\delta$-temporal motif counts that
correspond to two groups of motifs ($\delta =$ 1 hour).  Motifs
on the left capture ``blocking'' behavior, common in SMS messaging and Facebook wall
posting, and motifs on the right exhibit ``non-blocking'' behavior,
common in email.
}
\label{fig:blocking}
\end{figure}
\xhdr{Blocking communication}
If an individual typically waits for a reply from one individual before
proceeding to communicate with another individual, we consider it a
\emph{blocking} form of communication. A typical conversation between two
individuals characterized by fast exchanges happening back and forth is blocking
as it requires complete attention of both individuals. We capture this
behavior in the ``blocking motifs'' $M_{5,1}$, $M_{5,2}$ and $M_{6,2}$, which
contain 3 edges between two nodes with at least one edge in either direction
(Fig.~\ref{fig:blocking}, left).  However, if the reply doesn't arrive soon,
we might expect the individual to communicate with others without waiting for a
reply from the first individual. This is a non-blocking form of communication
and is captured by the ``non-blocking motifs'' $M_{4,1}$, $M_{4,3}$ and $M_{6,3}$
having edges originating from the same source but directed to different
destinations (Fig.~\ref{fig:blocking}, right)

The fractions of counts corresponding to the blocking and non-blocking motifs
out of the counts for all 36 motifs in Fig.~\ref{fig:three_edge_motifs}
uncover several interesting characteristics in communication networks ($\delta =
1$ hour; see Fig.~\ref{fig:blocking}). In $\fbwall$ and $\sms$, blocking
communication is vastly more common, while in $\emaileu$ non-blocking
communication is prevalent.  Email is not a dynamic method of
communication and replies within an hour are rare.  Thus, we would expect
non-blocking behavior.  Interestingly, the $\messages$ dataset shows both
behaviors as we might expect individuals to engage in multiple conversations
simultaneously.  In complete contrast, the $\phone$ dataset shows neither
behavior.  A simple explanation is that that a single edge (a phone call)
captures an entire conversation and hence blocking behavior does not emerge.

\begin{figure}[tb]
\centering
\includegraphics[width=\columnwidth]{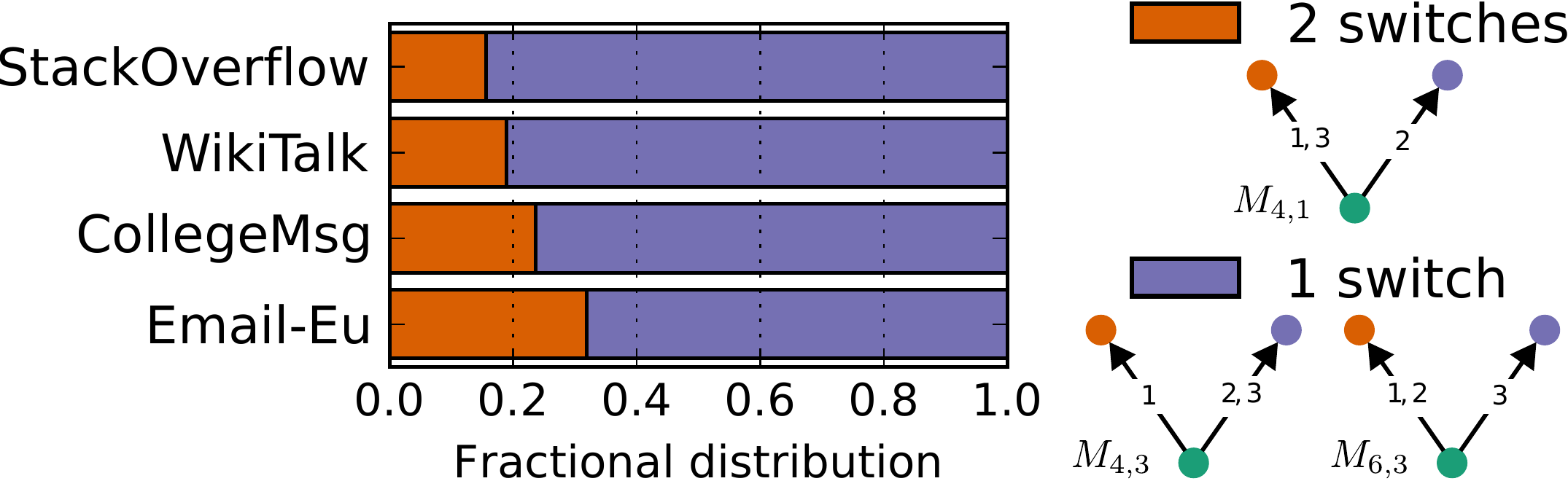}
\caption{%
Distribution of switching behavior amongst the non-blocking motifs.   Switching
is least common on Stack Overflow and most common in email.
}
\label{fig:switching}
\end{figure}
\xhdr{Cost of switching}
Amongst the non-blocking motifs discussed above, $M_{4,1}$ captures two
consecutive switches between pairs of nodes whereas $M_{4,3}$ and $M_{6,3}$ each
have a single switch (Fig.~\ref{fig:switching}, right). Prevalence of
$M_{4,1}$ indicates a lower cost of switching targets, whereas prevalence of the
other two motifs are indicative of a higher cost.  We observe in
Fig.~\ref{fig:switching} that the ratio of 2-switch to 1-switch motif counts
is the least in $\stackoverflow$, followed by $\wikitalk$, $\messages$ and then
$\emaileu$. On Stack Overflow and Wikipedia talk pages, there is a high
cost to switch targets because of peer engagement and depth of discussion.  On the other
hand, in the $\messages$ dataset there is a lesser cost to switch because it
lacks depth of discussion within the time frame of $\delta = $ 1 hour. Finally,
in $\emaileu$, there is almost no peer engagement and cost of switching is
negligible.

\begin{figure}[tb]
\centering
\includegraphics[width=1\columnwidth]{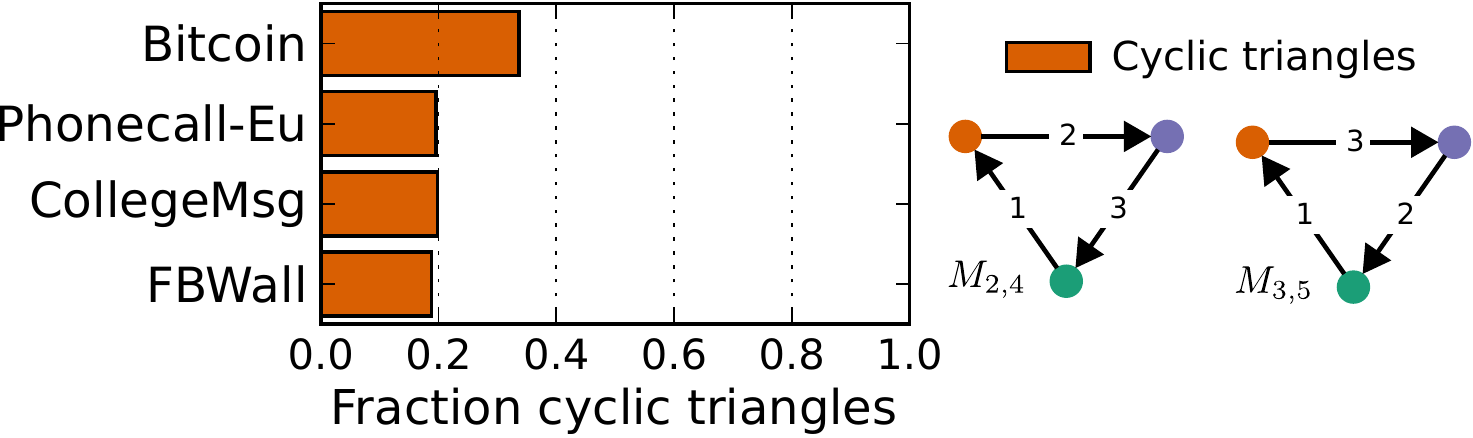}
\caption{%
Fraction of $3$-edge $\delta$-temporal triangle motif counts ($\delta =$ 1 hour) corresponding to
cyclic triangles (right) in the four datasets for which this fraction
is the largest.  $\bitcoin$ has a much higher fraction compared
to all other datasets. 
}
\label{fig:cycles}
\end{figure}

\xhdr{Cycles in $\bitcoin$}
Of the eight $3$-edge triangle motifs, $M_{2,4}$ and $M_{3,5}$ are cyclic, i.e.,
the target of each edge serves as the source of another edge.  We observe in
Fig.~\ref{fig:cycles} that the fraction of triangles that are cyclic is much higher in
$\bitcoin$ compared to any other dataset.
This can be attributed to the transactional nature of $\bitcoin$ where
the total amount of bitcoin is limited. Since remittance (outgoing edges) is typically associated
with earnings (incoming edges), we should expect cyclic behavior.

\begin{figure}[tb]
\centering
\includegraphics[width=0.95\columnwidth]{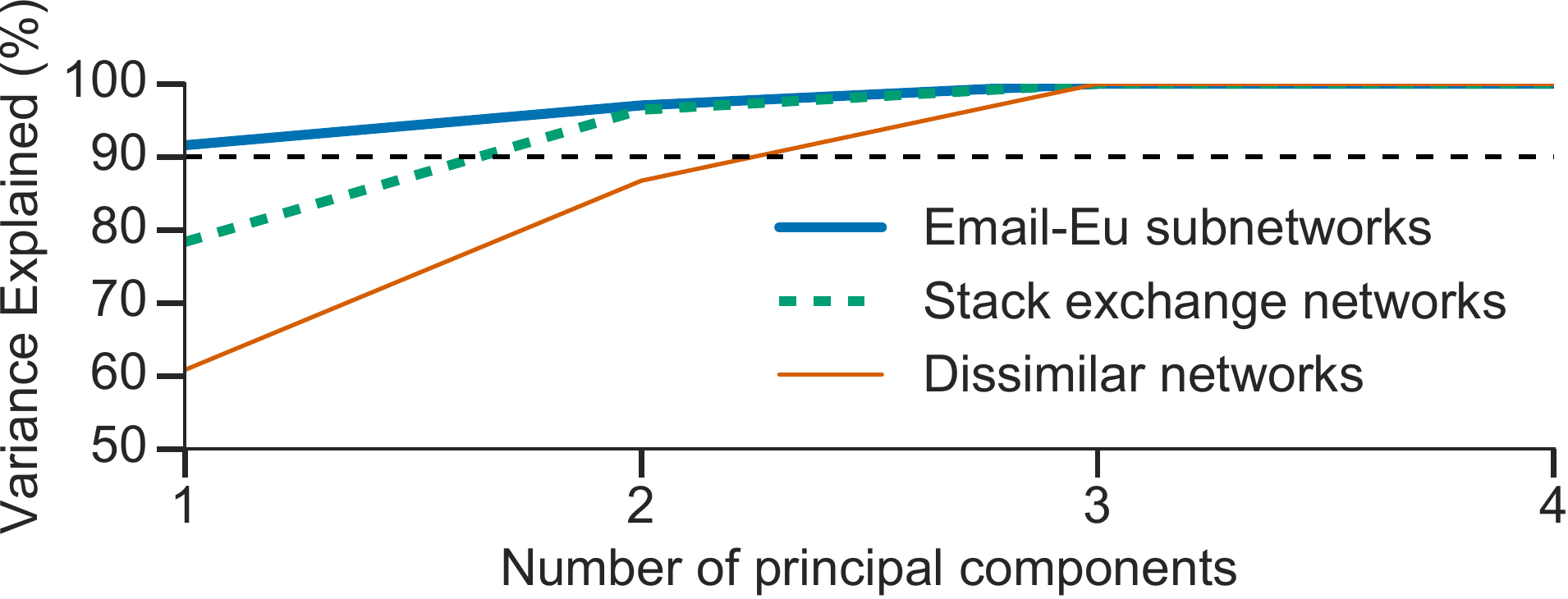}
\caption{%
Percentage of explained variance of relative counts of collections of datasets
plotted as a function of the number of principal components.  In datasets from
the same domain, 90\% of variance is explained with fewer components.
}
\label{fig:similar_counts_PCA}
\end{figure}

\xhdr{Datasets from the same domain have similar counts}
Static graphs from similar domains tend to have similar motif count distributions~\cite{milo2004superfamilies,vazquez2004topological,yaverouglu2014revealing}.
Here, we find similar results in temporal networks.  We formed two collections
of datasets from similar domains.  First, we took subsets of the $\emaileu$
dataset corresponding to email communication within four different departments
at the institution.  Second, we constructed temporal graphs from the stack
exchange communities Math Overflow, Super User, and Ask Ubuntu to study
in conjunction with the $\stackoverflow$ dataset.
We form count distributions by normalizing the counts of the 36 different
motifs in Fig.~\ref{fig:raw_counts}.
For datasets from a similar domain, we expect that if the count distributions 
are similar, then most of the variance is captured by a few principal components.
To compare, we use four datasets from dissimilar domains ($\emaileu$, $\phone$, $\sms$, $\wikitalk$). 
Fig.~\ref{fig:similar_counts_PCA} shows that to explain 90\% variance, $\emaileu$ subnetworks need just one principal component, stack exchange networks need two, and the dissimilar networks need three.

\begin{figure}[tb]
\centering
\includegraphics[width=0.9\columnwidth]{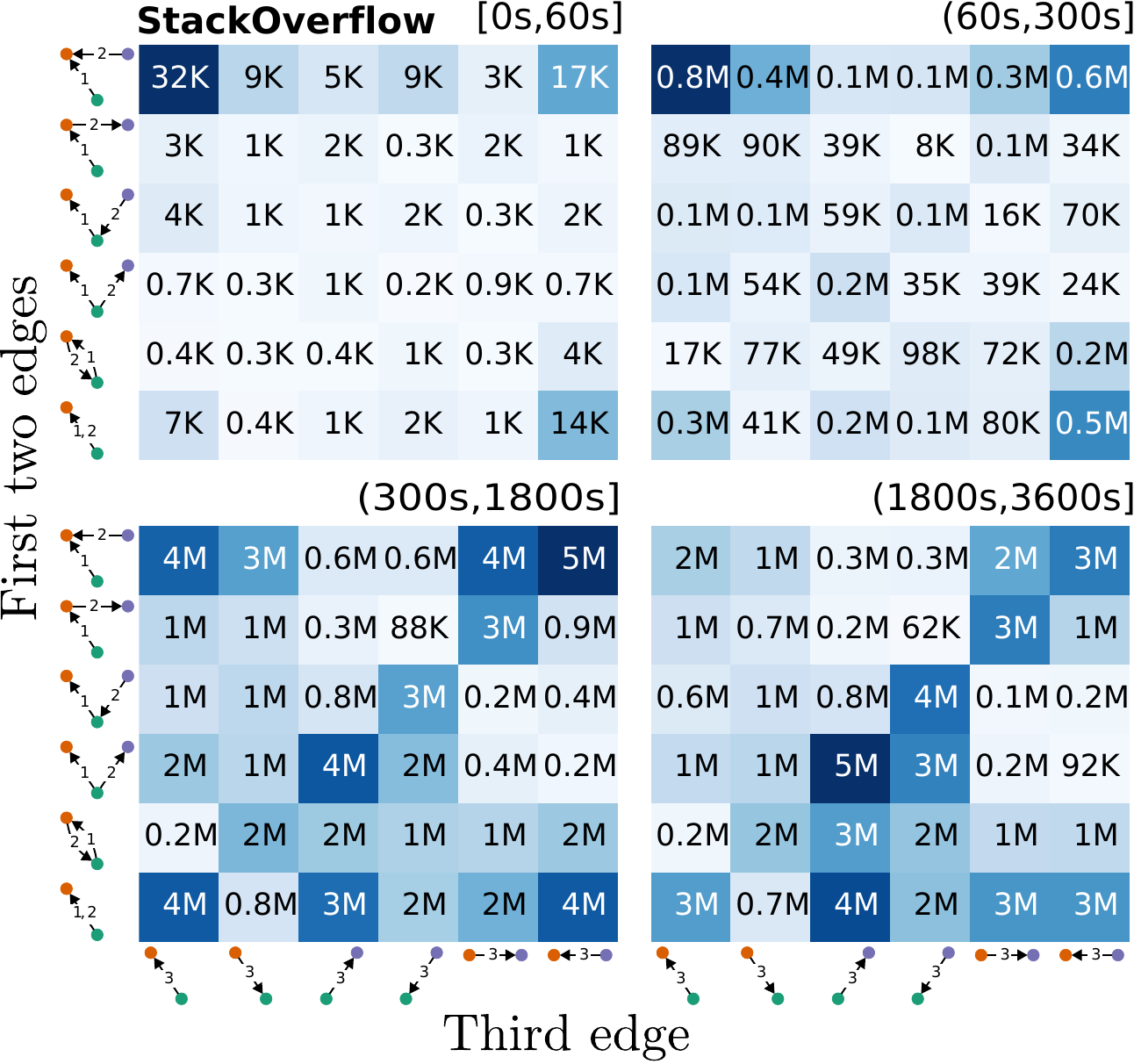}
\caption{%
Counts for all $2$- and $3$-node, $3$-edge $\delta$-temporal motifs in four time
intervals for the $\stackoverflow$ dataset.  For each
interval, the count in the $i$th row and $j$th column is the number of instances
of motif $M_{i, j}$ (see Fig.~\ref{fig:three_edge_motifs}).
}
\label{fig:so_counts_over_time}
\end{figure}

\xhdr{Motif counts at varying time scales} We now explore how motif counts change at different time scales.  For the
$\stackoverflow$ dataset we counted the number of instances of $2$- and
$3$-node, $3$-edge $\delta$-temporal motifs for $\delta = $ 60, 300, 1800, and
3600 seconds (Fig.~\ref{fig:so_counts_over_time}).  These counts determine the
number of motifs that completed in the intervals
[0, 60], (60, 300], (300, 1800s], and (1800, 3600] seconds
(e.g., subtracting 60 second counts from 300 second counts gives the interval (60, 300]).
Observations at smaller timescales reveal phenomenon which start to get eclipsed
at larger timescales.  For instance, on short time scales, motif $M_{1,1}$
(Fig.~\ref{fig:so_counts_over_time}, top-left corner) is quite common.  We
suspect this arises from multiple, quick comments on the original question, so
the original poster receives many incoming edges.  At larger time scales, this
behavior is still frequent but relatively less so.  Now let us compare counts
for $M_{1,5}$, $M_{1,6}$, $M_{2,5}$, $M_{2,6}$ (the four in the top right
corner) with counts for $M_{3,3}$, $M_{3,4}$, $M_{4,3}$, $M_{4,4}$ (the four in
the center). The former counts likely correspond to conversations with the
original poster while the latter are constructed by the same user interacting with
multiple questions.  Between 300 and 1800 seconds (5 to 30 minutes),
the former counts are relatively more common while the latter counts only become
more common after 1800 seconds. A possible explanation is that the typical
length of discussions on a post is about 30 minutes, and later on, users answer
other questions.

\begin{figure}[tb]
\centering
\vspace{-3mm}
\includegraphics[width=0.94\columnwidth]{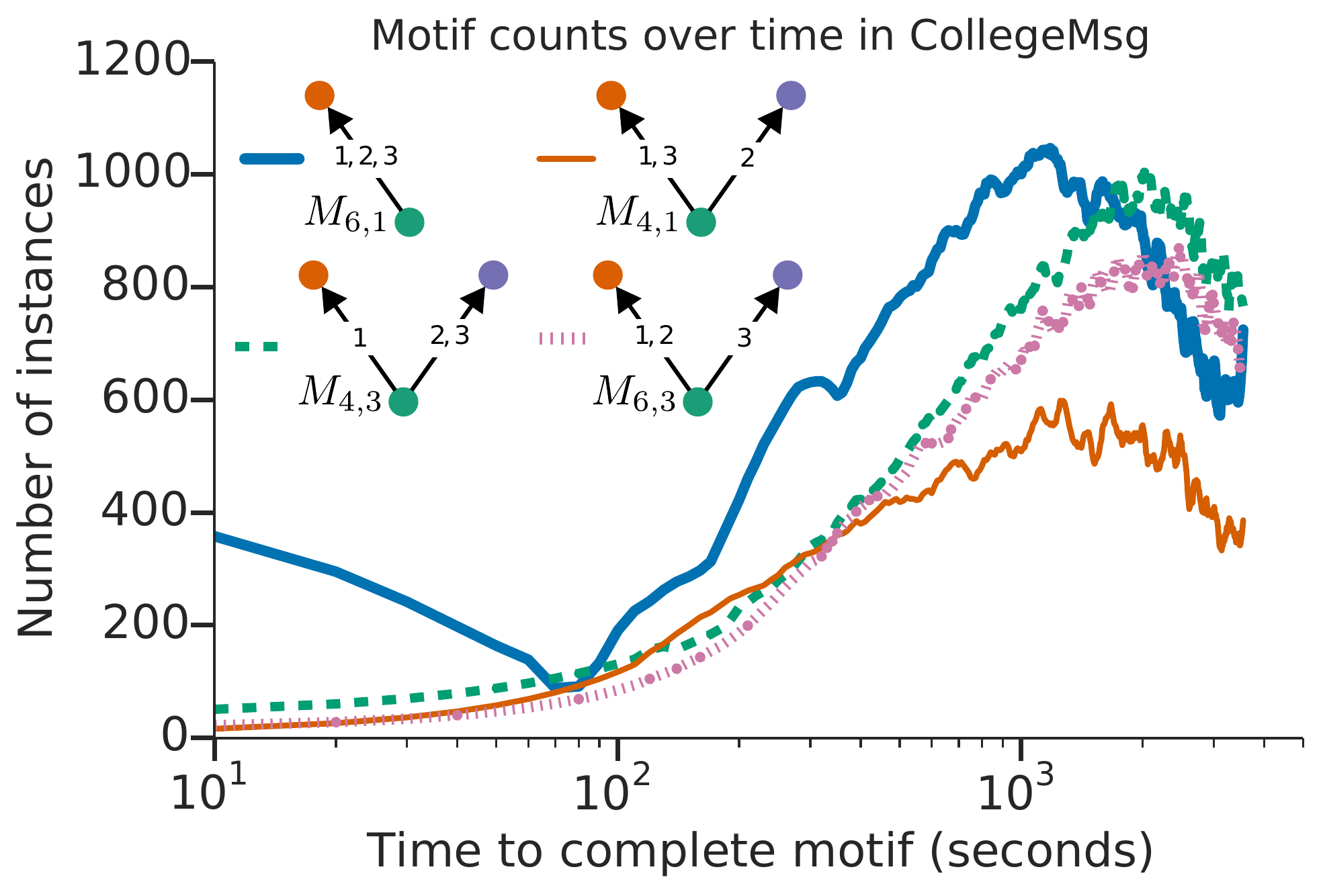}
\vspace{-1mm}
\caption{%
Counts over various time scales for the motifs representing a node sending $3$
outgoing messages to $1$ or $2$ neighbors in the $\messages$ dataset.
}
\label{fig:IM_over_time}
\end{figure}

Next, we examine messaging behavior in the $\messages$ dataset at fine-grained
time intervals.  We counted the number of motifs consisting of a single node
sending three outgoing messages to one or two neighbors (motifs $M_{6,1}$,
$M_{6,3}$, $M_{4,1}$, and $M_{4,3}$) in the time bins $[10(i-1), 10i)$ seconds,
$i = 1, \ldots, 500$ (Fig.~\ref{fig:IM_over_time}).  We first notice that at
small time scales, the motif consisting of three edges to a single neighbor
($M_{6,1}$) occurs frequently.  This pattern could emerge from a succession of
quick introductory messages.  Overall, motif counts increase from roughly 1
minute to 20 minutes and then decline.  Interestingly, after 5 minutes, counts
for the three motifs with one switch in the target ($M_{6,1}$, $M_{6,3}$, and
$M_{4,3}$) grow at a faster rate than the counts for the motif with two switches
($M_{4,1}$).  As mentioned above, this pattern could emerge from a tendency to
send several messages in one conversation before switching to a conversation
with another friend.

\subsection{Algorithm scalability}\label{sec:scalability}

\setlength{\tabcolsep}{2.5pt}
\begin{table}[tb]
\centering \caption{%
Time to count the eight 3-edge $\delta$-temporal triangle motifs
($\delta = 3600$) using the general counting method
(Alg.~\ref{alg:general}) and the fast counting method
(Alg.~\ref{alg:triangle}).  
}
  \begin{tabular}{l c c c c}
    \toprule 
    dataset & \# static                 & time, Alg.~\ref{alg:general} & time, Alg.~\ref{alg:triangle} & speedup \\
                 & triangles  & (seconds)                            & (seconds) & \\
    \midrule
    \wikitalk & 8.11M & 51.1 & 26.6 & 1.92x \\
    \bitcoin  & 73.1M & 27.3K & 483 & 56.5x \\
    \smsme & 78.0M & 2.54K & 1.11K & 2.28x \\
    \stackoverflow & 114M & 783 & 606 & 1.29x \\
    \phonecallme & 667M & 12.2K & 8.59K & 1.42x \\
    \bottomrule
    \end{tabular}\label{tab:scalability}
\end{table}

\begin{figure}[t]
\vspace{-2mm}
\centering
\includegraphics[width=0.94\columnwidth]{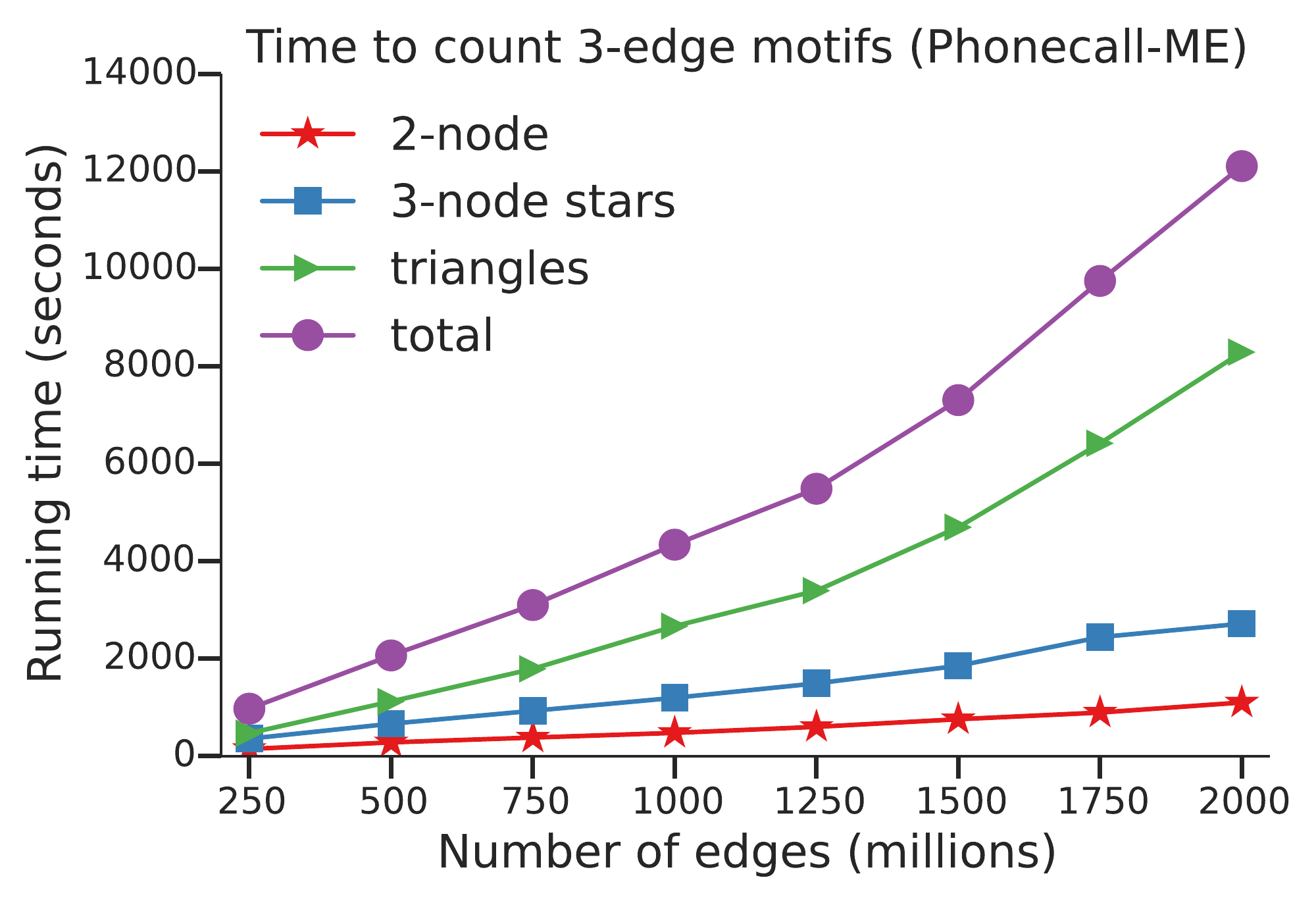}
\vspace{-1mm}
\caption{%
Time to count $3$-edge motifs on the first $k$ temporal million edges
in the $\phonecallme$ as a function of $k$.
}
\label{fig:scalability}
\end{figure}

Finally, we performed scalability experiments of our algorithms.  All
algorithms were implemented in C++, and all experiments ran using a single
thread of a 2.4GHz Intel Xeon E7-4870 processor.  We
did not measure the time to load datasets into memory, but our timings include
all pre-processing time needed by the algorithms (e.g., the triangle
counting algorithms first find triangles in the static graph).  We emphasize
that our implementation is single threaded, and the methods can be sped up
with a parallel algorithm.

First, we used both the general counting method (Alg.~\ref{alg:general})
and the fast counting method (Alg.~\ref{alg:triangle}) to count the number
of all eight 3-edge $\delta$-temporal triangle motifs in our datasets ($\delta =$ 1 hour).
Table~\ref{tab:scalability} reports the running times of the
algorithms for all datasets with at least one million triangles in the static
graph. For all of these datasets, our fast temporal triangle counting algorithm
provides significant performance gains over the general counting method, ranging
between a 1.29x and a 56.5x speedup.  The gains of the fast algorithm are the
largest for $\bitcoin$, which is due to some pairs of
nodes having many edges between them and also participating in many triangles.

Second, we measured the time to count various $3$-edge $\delta$-temporal motifs
in our largest dataset, $\phonecallme$.  Specifically, we measured the time to
compute (1) $2$-node motifs, (2) $3$-node stars, and (3) triangles on the first
$k$ million edges in the dataset for $k = 250, 500, \ldots, 2000$
(Fig.~\ref{fig:scalability}).  The time to compute the $2$-node, $3$-edge
motifs and the $3$-node, $3$-edge stars scales linearly, as expected from our
algorithm analysis.  The time to count triangle motifs grows superlinearly and
becomes the dominant cost when there is a large number of edges.  For practical
purposes, the running times are quite modest.  With two billion edges, our
methods take less than 3.5 hours to complete (executing sequentially).

\section{Discussion}
\label{sec:discussion}

We have developed $\delta$-temporal network motifs as a tool for analyzing
temporal networks.  We introduced a general framework for counting
instances of any temporal motif as well as faster algorithms for
certain classes of motifs and found that motif counts reveal key structural
patterns in a variety of temporal network datasets.
Our work opens a number of avenues for additional research.  First, our fast
algorithms are designed for $3$-node, $3$-edge star and triangle motifs.  We
expect that the same general techniques can be used to count more complex
temporal motifs.  Next, it is important to note that our fast algorithms
only \emph{count} the number of instances of motifs rather than \emph{enumerate}
the instances.  This concept has also been used to accelerate static motif
counting~\cite{pinar2016escape}.  Temporal motif enumeration algorithms provide
an additional algorithmic design challenge.  There is also a host of theoretical
questions in this area for lower bounds on temporal motif counting.  Finally,
motif counts can also be measured with respect to a null
model~\cite{kovanen2011temporal,milo2002network}.  Such analysis may yield
additional discoveries.  Importantly, our algorithms will speed up such
computations, which use raw counts from many random instances of a generative
null model.

\xhdr{Acknowledgements}
We thank Moses Charikar for valuable discussion.
This research has been supported in part by NSF
IIS-1149837, ARO MURI,
DARPA SIMPLEX and NGS2,
Boeing,
Bosch,
Huawei,
Lightspeed,
SAP,
Tencent,
Volkswagen,
Stanford Data Science Initiative,
and a Stanford Graduate Fellowship.

\vspace{-0.25cm}
\bibliographystyle{abbrv}
\bibliography{refs}


\end{document}